\documentclass[11pt, a4paper]{amsart}

\usepackage[latin1]{inputenc}
\usepackage[all]{xy}
\usepackage{amssymb, amsmath, amscd, geometry,mathtools}
\usepackage{braket}
\DeclarePairedDelimiter{\abs}{\lvert}{\rvert}
\usepackage{graphicx}

\numberwithin{equation}{section}

\usepackage{latexsym}
\usepackage[usenames]{color}
\usepackage{epic} %%%% Circunferencias grandes y curvas de bezier
\usepackage{mathrsfs}
\usepackage{euscript}
\usepackage{rotating}
\usepackage{bbm}
 \usepackage{url}

\usepackage{verbatim}
\usepackage[usenames]{color}

\theoremstyle{plain}
\newtheorem{theo}{Theorem}[section]

\newtheorem{definition}[theo]{Definition}
\newtheorem{proposition}[theo]{Proposition}
\newtheorem{theorem}[theo]{Theorem}

\newtheorem{lemma}[theo]{Lemma}
\theoremstyle{definition}

\newtheorem{remark}[theo]{\textbf{Remark}}

\newcommand\numeq[1]%
  {\stackrel{\scriptscriptstyle(\mkern-1.5mu#1\mkern-1.5mu)}{=}}
  
 \newcommand\numineq[1]%
  {\stackrel{\scriptscriptstyle(\mkern-1.5mu#1\mkern-1.5mu)}{\leq}}

      \newcommand{\R}{{\mathbb R}}
      
\newcommand{\id}{\operatorname{id}}

\begin{document}

\author{}

\title[Optimal non-signalling violations via Tensor Norms]{Optimal non-signalling violations via Tensor Norms}

\author{A. Amr}
\author{C. Palazuelos}
\author{I. Villanueva}

\thanks{}

\begin{abstract}
In this paper we characterize  the set of bipartite non-signalling probability distributions in terms of tensor norms. Using this characterization we give optimal upper and lower bounds on Bell inequality violations when non-signalling distributions are considered. Interestingly, our upper bounds show that non-signalling Bell inequality violations cannot be significantly larger than quantum Bell inequality violations.
\end{abstract}

\maketitle

\section{Introduction}

A very remarkable feature of quantum mechanics is that it predicts the existence of experimental data which cannot be reproduced within any local and realistic physical theory, even in the presence of hidden variables. This idea was first formalized by Bell \cite{Bell} and has played a major role in the recent development of quantum information science (see the survey \cite{BCPSW}). 

One of the main ideas in Bell's work, non-locality, can be studied in itself, independently from  quantum mechanics. Bell's scenario is usually described by  two parties spatially separated, typically named Alice and Bob, who perform different measurements to obtain certain outputs. If we label Alice's and Bob's measurement devices by $x$ and $y$ respectively so that $x,y=1,\cdots , N$ and Alice's and Bob's possible outputs by $a$ and $b$ respectively so that $a,b=1,\cdots ,K$, the main object of study is the tensor $$P=\{P(a,b|x,y)\}_{x,y; a,b=1}^{N,K},$$ where $P(a,b|x,y)$ denotes the probability that Alice and Bob obtain the pair of outputs $(a,b)$ when they measure with the inputs $x$ and $y$ respectively. Note that, from an algebraic point of view, each $P$ is just an element in $\R^{N^2K^2}$. Moreover, the fact that $P$ describes a measurement scenario implies that certain restrictions must be fulfilled; namely,  $P(a,b|x,y)\geq 0$ and $\sum_{a,b}P(a,b|x,y)=1$ for every $a,b,x,y$. Let us denote by $\mathcal{C}$ the subset of $\R^{N^2K^2}$ given by such elements. We  will refer to them as  \emph{probability distributions}.

The main point in Bell's work was to understand that the assumption of a physical theory to explain the experiment (and, more generally, Nature) leads to a subset of $\mathcal{C}$ which will be formed by those probability distributions which are compatible with such a theory. A minimal requirement for a theory to be meaningful is the so called \emph{non-signalling condition}, which means that Alice and Bob's marginals are well defined:
\begin{align}
	\sum_{a}P(a,b|x,y)=\sum_{a}P(a,b|x',y)\text{ for all } x, x', y, b\label{ns1},\\
	\sum_{b}P(a,b|x,y)=\sum_{b}P(a,b|x,y')\text{ for all } y, y', x, a\label{ns2}.
\end{align}

This is physically motivated by the principle of \emph{Einstein locality}, which implies non-signalling if we assume that Alice and Bob are space-like separated. Let us denote the set of non-signalling probability distributions by $\mathcal{NS}\subset \mathcal{C}$. 

Among all subsets of $\mathcal{NS}$ there are two which are particularly relevant. They are the set of probability distributions which are compatible with a classical description and a quantum description of Nature respectively. More precisely,  given a probability distribution $P\in \mathcal{C}$, we will say that $P$ is \emph{Classical} if
\begin{equation*}\label{classical}
P(a,b|x,y)=\int_\Omega P_\omega(a|x)Q_\omega(b|y)d\mathbb{P}(\omega) \hspace{0.4 cm}\text{for every $x,y,a,b$,}
\end{equation*}
where $(\Omega,\Sigma,\mathbb{P})$ is a probability space, $P_\omega(a|x)\ge 0$ and $\sum_a
P_\omega(a|x)=1$ for all $a, x,\omega$; and analogous conditions hold for $Q_\omega(b|y)$. On the other hand, we say that $P$ is \emph{Quantum} if there exist two Hilbert spaces $H_1$, $H_2$ such that
\begin{equation*}\label{quantum}
P(a,b|x,y)=\langle\psi|E_x^a\otimes F_y^b |\psi\rangle\hspace{0.4 cm}\text{for every $x,y,a,b$,}
\end{equation*}
where  $|\psi\rangle\in H_1\otimes H_2$ is a vector of norm one and $(E_x^a)_{x,a}\subset B(H_1)$, $(F_y^b)_{y,b}\subset B(H_2)$ are two sets of operators representing POVM measurements on Alice's and Bob's system respectively. That is, $E_x^a$ is semidefinite positive and $\sum_{a}E_x^a=\id_{H_1}$ for every $a$, $x$; and analogous conditions hold for $(F_y^b)_{y,b}$.  We will denote by  $\mathcal{L}$  and $\mathcal{Q}$ the sets of classical and quantum probability distributions respectively. It is well known \cite{Tsirelson} that $\mathcal{L}\varsubsetneq\mathcal{Q}\varsubsetneq\mathcal{NS}$, being the three of them convex sets.

A natural quantification of how different the sets  $\mathcal L$ , $\mathcal Q$ and  $\mathcal {NS}$  are can be done by means of the so called \emph{Bell inequality violations}. More precisely, if $\mathcal A\in \{\mathcal L, \mathcal Q, \mathcal {NS}\}$ and $M\in \R^{N^2K^2}$ is any tensor, let us denote
\begin{align*}
\omega_{\mathcal A} (M)=\sup_{P\in \mathcal A} |\langle M,P\rangle|,
\end{align*}where the dual action is given by $\langle M,P\rangle=\sum_{x,y;a,b=1}^{N,K}M_{x,y}^{a,b}P(a,b|x,y)$. Then, we define the quantities\footnote{Observe that both quantities depend on $N$ and $K$, so we should denote $LV^{N,K}_{\mathcal Q}$ and $LV^{N,K}_{\mathcal{NS}}$, but we will simplify notation when $N$ and $K$ are clear from the context.}
\begin{align*}
LV_{\mathcal Q}=\sup_M\frac{\omega_{\mathcal Q}(M)}{\omega_{\mathcal L}(M)}, \hspace{0.4 cm} \text{ and }\hspace{0.4 cm} LV_{\mathcal{NS}}=\sup_M\frac{\omega_{NS}(M)}{\omega_{\mathcal L}(M)}.
\end{align*}

The quantity $LV_{\mathcal Q}$ has been deeply studied during the last years. The fact that $LV_{\mathcal Q}>1$ is rephrased as the existence of quantum probability distributions which are not classical (phenomenon known as \emph{quantum no-locality}) and $LV_{\mathcal Q}$ can be understood as a measure of the deviation of quantum mechanics from the classical theory. Because of historical reasons we sometimes denote a general tensor $M$ as a Bell inequality\footnote{Formally, the tensor $M$ defines the inequality $\langle M,P\rangle\leq \omega_{\mathcal L}(M)$ for every $P\in  \mathcal L$.}. Beyond its theoretical interest as a measure of non-locality, $LV_{\mathcal{\mathcal Q}}$ turns out to be very useful in many different tasks such as dimension witness, communication complexity or entangled games (see \cite{JPPVW2} for more detail).

The aim of this work is to focus on the quantity $LV_{\mathcal{NS}}$ as a way to study the ultimate limitations of any \emph{meaningful} physical theory.  Indeed, since every set of probability distributions defined by a physical theory via a Bell scenario (when Alice and Bob are space-like separated and there is no communication between them) must be contained in $\mathcal{NS}$, the quantity $LV_{\mathcal{NS}}$ should be understood as a global upper bound of the most extreme behaviour  (in terms of Bell violations) we can expect from any such theory. In addition, since non-locality is behind some of the most important applications of quantum information such as device-independent quantum cryptography (see for instance \cite{ABGMPS}), where one is particularly interested in avoiding hypothesis about the adversaries (such as being quantum), the quantity $LV_{\mathcal{NS}}$ can be thought of an abstract quantification of the most extreme possible scenario one can expect.

In addition, Bell inequality violations have a particularly interesting interpretation when we look at two-prover one-round games (or, simply, games), where two collaborative players must answer some outputs after being asked certain questions by a referee. These games play a major role in computer science because many interesting problems can be re-phrased in terms of them. The interesting point for us is that games can be identified with particular tensors $G\in \R^{N^2K^2}$ with non-negative coefficients. In that context,  the quantity $\omega_{\mathcal A} (G)$ denotes the highest probability of winning the corresponding game when the players are restricted to the use of classical resources ($\mathcal A=\mathcal L$), quantum resources ($\mathcal A=\mathcal Q$) and non-signalling resources ($\mathcal A=\mathcal NS$) in their strategies to play the game. Then, the quantities $LV_{\mathcal Q}$ and $LV_{\mathcal{NS}}$ restricted to games $G$ are measures of how much better quantum  and non-signalling strategies are compared to classical ones. 

In order to study the quantity $LV_{\mathcal{NS}}$ we will follow a similar approach to the one followed by Tsirelson in his seminal paper \cite{TsirelsonII} to study quantum correlation matrices. That is, we will understand this problem by means of tensor norms and then we will crucially use Banach space machinery to study it. This was also the spirit of the works \cite{JungeP11low, JPPVW, JPPVW2} to study $LV_{\mathcal Q}$, where the authors reduced the problem to the study of two different norms in $\R^{N^2K^2}=\R^{NK}\otimes \R^{NK}$. Indeed, based on \cite{JungeP11low, JPPVW},  it was shown in \cite{PV-Survey} that, given a game $G$, one has
\begin{align}\label{classical-quantum games}
\omega_{\mathcal L}(G)= \|G\|_{\ell_1^N(\ell_\infty^K)\otimes_{\epsilon} \ell_1^N(\ell_\infty^K)},\hspace{0.4 cm}\text{   and  }\hspace{ 0.4 cm  }\omega_{\mathcal{Q}}(G)= \|G\|_{\ell_1^N(\ell_\infty^K)\otimes_{min} \ell_1^N(\ell_\infty^K)},
\end{align}where here the first norm is the $\epsilon$ tensor norm in the category of Banach spaces and the second norm is the minimal tensor norm in the category of operator spaces. A precise definition of these spaces, as well as those appearing in Equations (\ref{NSnormBell}) and (\ref{NS norm}) below, will be given in Section \ref{Sec:Preliminaries}. Using this description of the quantities $\omega_{\mathcal L}(G)$ and $\omega_{\mathcal Q}(G)$, one can easily show \cite{JungeP11low, PV-Survey}  that 
\begin{align}\label{upper bound quantum}
LV_{\mathcal Q}\leq O(\text{min}\{N,K\}),
\end{align}when we restrict to games. In fact, as we will explain later, this upper bound also works for general Bell inequalities.

It turns out that the non-signalling value of a game can also be described by  a certain natural norm in $\R^{N^2K^2}$.  However, the description in the new context is trickier because of the absence of a tensor product structure. Indeed, it can be seen that no tensor norm in $\ell_1^N(\ell_\infty^K)\otimes  \ell_1^N(\ell_\infty^K)$ can describe the non-signalling value of a game. Instead, given a general tensor $M\in \R^{N^2K^2}$, one has to consider the following norm:
\begin{align}\label{NSnormBell}
\|M\|_{\text{DNS}}=\inf\big\{\|M_1\|_{\ell_1^N(\ell_\infty^K(\ell_1^N(\ell_\infty^K)))} +\|M_2^T\|_{\ell_1^N(\ell_\infty^K(\ell_1^N(\ell_\infty^K)))} :\, M=M_1+M_2\big\},
\end{align} where, for a given $z\in \R^N\otimes \R^K\otimes \R^N\otimes\R^K$, $z^T=\text{flip}(z)$ and  $\text{flip}: \mathbb R^{N^2K^2}\longrightarrow \mathbb R^{N^2K^2}$ is the linear map  defined on elementary tensors by $\text{flip}(e_x\otimes e_a\otimes e_y\otimes e_b)=e_y\otimes e_b\otimes e_x\otimes e_a$. Then, our first result provides a natural description for the non-signalling value of a game.
\begin{theorem}\label{NS value games}
Given a tensor $G\in \R^{N^2K^2}$ with non-negative coefficients, we have
$$\omega_{\mathcal{NS}}(G)=\|G\|_{\text{DNS}}.$$
\end{theorem}

With this description at hand, one can show the following upper bound. 
\begin{theorem}\label{thm: games}
Given a tensor $G\in \R^{N^2K^2}$ with non-negative coefficients,
$$\frac{\omega_{\mathcal{NS}}(G)}{\omega_{\mathcal L}(G)}\leq \text{min}\{N,K\}.$$ 
\end{theorem}

Hence, somehow surprisingly, we see that, although a priori non-signalling strategies can be much better than quantum strategies, the same upper bound applies in both cases. In fact, it is known that the upper bound (\ref{upper bound quantum}) is optimal in the number of outputs $K$ (\cite{BuhrmanRSW12}). Hence, in terms of this number, non-signalling strategies do not provide an  advantage with respect to quantum strategies.  In other words, when comparing with classical strategies, quantum strategies are as good as any other non-signalling theory can be.

Interestingly, our second result shows that Theorem \ref{thm: games} is  in fact essentially optimal in both parameters.
\begin{theorem}\label{thm: games lower bound}
For every natural number $n$ there exists a pointwise non-negative tensor $G_n\in \R^{N^2K^2}$, with $N=K=n$,  such that
$$\frac{\omega_{\mathcal{NS}}(G)}{\omega_{\mathcal L}(G)}\geq D\frac{n}{\log n},$$where $D$ is a universal constant.
\end{theorem}

Although pointwise non-negative tensors, in particular games, have a great relevance in computer sciences, from a physical point of view, studying this particular case is not enough to have a good knowledge about the sets $\mathcal L$, $\mathcal Q$ and $\mathcal {NS}$. More precisely, one can easily check (see \cite[Section 5]{JungeP11low} for a more complete study of the geometry of these sets) that if we define $$\widetilde{\mathcal A}=co\big(\mathcal A\cup -\mathcal A\big)\hspace{0.4 cm}\text{ for  }\hspace{0.4 cm}\text{$\mathcal A\in \{\mathcal L, \mathcal Q, \mathcal {NS}\}$},$$ where  $\text{co}(X)$ denotes the convex hull of $X$, the quantities $LV_{\mathcal{Q}}$ and $LV_{\mathcal{NS}}$ are the smallest positive numbers such that $$\widetilde{\mathcal Q}\subseteq LV_{\mathcal{Q}}\cdot \widetilde{\mathcal L}\hspace{0.4 cm}\text{ and  }\hspace{0.4 cm}\widetilde{\mathcal NS}\subseteq LV_{\mathcal{NS}}\cdot \widetilde{\mathcal L}.$$

Observe that the upper bound in Theorem \ref{thm: games} for non-negative tensors does not say anything about the values $LV_{\mathcal{Q}}$ and $LV_{\mathcal{NS}}$. This can be easily understood by looking at the context where we fix $K=2$. While, the upper bound (\ref{upper bound quantum}) applies in this context to state that $LV_{\mathcal Q}(N,2)=O(1)$ independently of $N$, it is well known (and we will explain it later) than $LV_{\mathcal {NS}}(N,2)= \Omega(\sqrt{N})$ in this case. In particular, we know that Theorem \ref{thm: games} cannot hold for general elements in $\R^{N^2K^2}$.

On the other hand, the study of general Bell inequalities presents some additional problems. As it was explained in \cite[Section V]{JungeP11low}  (see also \cite[Section IV]{PV-Survey}), the main issue in the general context is that, since the set $\mathcal C$ is contained in a proper affine subspace of $\R^{N^2K^2}$, one cannot expect the values $\omega_{\mathcal A}$ to be equivalent to certain norms in $\R^{N^2K^2}$. To circumvent this problem, in \cite{JungeP11low} the authors defined an alternative (operator) space denoted by $NSG(N,K)^*$ of dimension $NK-N+1$, such that Equation (\ref{classical-quantum games}) holds for every tensor $M$ when the space $\ell_1^N(\ell_\infty^K)$ is replaced by $NSG(N,K)^*$. This result provided a geometrical description of the sets $\mathcal L$ and $\mathcal Q$ as unit balls of certain well known tensor norms. In addition, it was shown that the space $NSG(N,K)^*$ is a ``twisted version'' of the space $\ell_1^N(\ell_\infty^K)$, allowing the authors to prove the upper bound in Equation  (\ref{upper bound quantum}) for general tensors. As we mentioned before, the non-signalling case is different because splitting the space $\R^{N^2K^2}$ as $\R^{NK}\otimes \R^{NK}$ does not seem to be useful. In fact, we will show that the right Banach space to be considered in the new context is the real linear space $\mathcal{ANS}$, defined by the elements $R\in \R^{N^2K^2}$ for which there exist two tensors $P, Q\in \R^{NK}$ and a constant $c\in \R$ verifying
\begin{align*}
&\sum_pR_{xy}^{a,p}=P_x^a, \hspace{0.3 cm} \sum_qR_{xy}^{q,b}=Q_y^b\hspace{0.3 cm}\text{ and}\hspace{0.3 cm} \sum_{p,q}R_{xy}^{p,q}=c \hspace{0.4 cm}\text{for every $x,y,a,b$},
\end{align*}
endowed with the norm
\begin{align}\label{NS norm}
\|R\|_{\text{NS}}=\max\Big\{\|R\|_{\ell_\infty^N(\ell_1^K(\ell_\infty^N(\ell_1^K)))},\|R^T\|_{\ell_\infty^N(\ell_1^K(\ell_\infty^N(\ell_1^K)))}\Big\}.
\end{align}

This Banach space allows us to completely characterize the set of non-signalling probability distributions by means of a natural norm.

\begin{theorem}\label{NS description}
Let $\mathcal{ANS}$ be the linear space above endowed with the norm $\|\cdot\|_{\text{NS}}$ and let us denote by $B_{\mathcal{ANS}}$ its unit ball. Then,
$$B_{\mathcal{ANS}}=co(\mathcal{NS}\cup -\mathcal{NS}).$$
\end{theorem}

In addition, we will show in Theorem \ref{banachmazurdistance} that the space $\mathcal{ANS}$ is a twisted version of  $\R^{N^2K^2}$ endowed with the norm $\|\cdot\|_{\text{NS}}$ in the corresponding dimension. As a consequence of this, we can use  techniques from Banach space theory to obtain the following upper bound.

\begin{theorem}\label{boundlv}
Let $N$ and $K$ be two natural numbers. Then, 
$$LV_{\mathcal{NS}}\leq O\big(\text{min}\{N, \sqrt{NK}\}\big).$$

Moreover, according to Theorem \ref{thm: games lower bound}, this upper bound is sharp and it is attained on non-negative tensors.
\end{theorem}

As we have commented before, one cannot expect to have un upper bound similar to (\ref{upper bound quantum}) for the $LV_{\mathcal{NS}}$ because no upper bound for this second quantity can depend only on the number of outputs. However, we see that the upper bound we obtain is ``morally'' comparable to the one for the quantum value of Bell inequalities. Since it is not known if the upper bound (\ref{upper bound quantum}) in $N$ is attained for quantum probability distributions (the best result we have so far is $LV_{\mathcal{Q}}=\Omega(\sqrt{N})$, proved in \cite{JungeP11low}) we cannot conclude that the largest non-signalling Bell violation is comparable to the largest quantum Bell violation. However, our bounds show that this result is not only possible, but very plausible. This emphasizes the idea that, in some sense, the theory of quantum mechanic is as non-local as any other physical theory can be.

\

The organization of the rest of the paper is the following. In Section \ref{Sec:Preliminaries} we introduce some notions about  tensor norms on Banach spaces and we will analyze the particular case of correlation Bell inequalities, showing how in this context the unit ball of different tensor norms precisely describe the different sets of correlations. In Section \ref{Sec:NSnorm} we will study the right norm to be considered in the case of non-signalling probability distributions and also its dual norm, which will lead us to the study of games. In particular, we will prove Theorem \ref{NS value games} and Theorem \ref{thm: games}. In Section \ref{Sec:Game} we will show Theorem \ref{thm: games lower bound} by showing the existence of a family of games $G_n$ for which the non-signalling value and the classical value give a gap of order $n/\log n$. Finally, in Section \ref{Sec: General Bell} we will deal with the case of general tensors and we will prove Theorem \ref{NS description} and Theorem \ref{boundlv}.

\section{Preliminaries}\label{Sec:Preliminaries}

In this section we  introduce  basic notions from  Banach space theory which we will later need.  

Given a normed space $X$, $B_X=\{x\in X\text{ such that }\|x\|\leq1\}$ is its unit ball. The dual space consisting of linear and continuous maps from $X$ to the scalar field $\mathbb{R}$ will be denoted by $X^*$ and its norm has the natural expression $\|x^*\|_{X^*}=\sup_{x\in B_X}\abs{\braket{x^*|x}}$.

All the Banach spaces we are considering in this article are finite dimensional. In particular we will be very interested in the spaces $\ell_1^N$ and $\ell_\infty^N$, enhanced with other Banach spaces, lets say $X$, to create $\ell_1^N(X)$ and $\ell_\infty^N(X)$. The definitions of these two spaces are sequences of $N$ elements in $X$ and the norm of an element $u=\{x_i\}_{i=1}^N$ with $x_i\in X$ is the following:
\begin{align*}
&\|u\|_{\ell_1^N(X)}=\sum_{i=1}^N\|x_i\|_X, \\
&\|u\|_{\ell_\infty^N(X)}=\max_{1\leq i \leq N} \|x_i\|_X.
\end{align*} 

Whenever we have two  finite dimensional normed spaces $X$ and $Y$ we can consider the tensor product of them $X\otimes Y$ and endow it with different tensor norms in order to define different Banach spaces (see \cite[Section 2]{defantfloret} for the following definitions and relations). For a given $u\in X\otimes Y$ the $\epsilon$-norm and the $\pi$-norm are defined by:

\begin{align}\label{Def_eps_pi}
\|u\|_{X\otimes_\epsilon Y}&=\sup\Big\{\abs{\braket{u,x^*\otimes y^*}}: x^*\in B_{X^*},y^*\in B_{Y^*}	\Big\},\\
\nonumber\|u\|_{X\otimes_\pi Y}&=\inf\Big\{\sum_{i=1}^N \|x_i\|_X\|y_i\|_Y: N\in \mathbb N \text{   }\text{ and }\text{  }u=\sum_{i=1}^N x_i\otimes y_i\Big\}.
\end{align}

We will use the notation $X\otimes_\epsilon Y$ and $X\otimes_\pi Y$ to refer to the space $X\otimes Y$ endowed with each of the previous norms. The $\epsilon$- and $\pi$-norm are dual to each other (in finite dimensional spaces): 

$$(X\otimes_\epsilon Y)^*=X^*\otimes_\pi Y^* \mbox{ and } (X\otimes_\pi Y)^*=X^*\otimes_\epsilon Y^* \hspace{0.5 cm}(isometrically).$$ 

It follows from the definitions that $\ell_1^N(X)=\ell_1^N\otimes_\pi X$ and also $\ell_\infty^N(X)=\ell_\infty^N\otimes_\epsilon X$. 

In particular, for a given $z=\{z(x,a,y,b)\}_{x,a,y,b}\in \R^N\otimes \R^K\otimes \R^N\otimes\R^K$, we can define:
\begin{align*}
\|z\|_{\ell_\infty^N(\ell_1^K(\ell_\infty^N(\ell_1^K)))}&=\max_x\sum_a\max_y\sum_b \abs{z(x,a,y,b)},\\
\|z\|_{\ell_1^N(\ell_\infty^K(\ell_1^N(\ell_\infty^K)))}&=\sum_x\max_a\sum_y\max_b\abs{z(x,a,y,b)}.
\end{align*}

For two isomorphic Banach spaces $X$ and $Y$ we can define the Banach-Mazur distance as $d(X,Y)=\inf \{\|T\|\|T^{-1}\|\text{ such that }T \text{ is an isomorphism from }X\text{ to }Y \}$\cite{tomczak-jaegermann}.

Finally, given two Banach spaces $X$ and $Y$ that are subspaces of some other Banach space $Z$, we can consider two more spaces $X\cap Y$ and $X+Y$ with the following norms: 
\begin{align}\label{sum-int}
\|x\|_{X\cap Y}&=\max\{\|x\|_X,\|x\|_Y\},\\
\nonumber \|x\|_{X+Y}&=\inf\{\|x_1\|_X+\|x_2\|_Y\text{ such that }x=x_1+x_2\}.
\end{align}

It is not hard to see  \cite[Chapter 2]{interpolation} that if $X\cap Y$ is dense in both $X$ and $Y$, then $(X\cap Y)^*=X^*+Y^*$. In the case we will be interested, $X$ and $Y$ will be  finite dimensional with the same dimension. Since all norms on a finite dimensional space are equivalent, we can consider $Z$ to be either $X$ or $Y$ and $X\cap Y$ will be not only dense in $X$ and $Y$, but actually will coincide (as a vector space) with both of them. 

The study of probability distributions with only two possible outputs, let us say $\{+1,-1\}$, is specially relevant. In that case $P\in\mathbb{R}^{4N^2}$, and it becomes interesting to work, not with the full probability distribution, but with the following correlations:
$$\gamma=(\gamma_{xy})_{x,y=1}^N\in\mathbb{R}^{N^2}\text{ where }\gamma_{xy}=\mathbb{E}[P(a\cdot b|x,y)]=P(a{\cdot}b{=}{1}|x,y)-P(a{\cdot}b{=}{-}1|x,y).$$

If the correlations are generated from classical probability distributions, then we call them classical correlations, and we denote its subset as $\mathcal{L}_c$. The non-signalling case will be denoted as $\mathcal{NS}_c$. It is known that a non-signalling distribution is uniquely determined by the expected correlations and the expected marginals, defined as $M_A(x)=\mathbb{E}[P(a|x)]$ and $M_B(y)=\mathbb{E}[P(b|y)]$ \cite[Proposition 1]{communicationcomplexityns}.

Correlations can be characterized in terms of tensor norms. It follows  easily from the definitions of $\epsilon$ and $\pi$ norm (see \cite[Proposition 1,2]{communicationcomplexityns})
 that:
\begin{align*}
&\gamma\in\mathcal{L}_c\text{   }\text{ if and only if }\text{   }\|\gamma\|_{\ell_\infty^N\otimes_\pi\ell_\infty^N}\leq1,\\
&\gamma\in\mathcal{NS}_c\text{   }\text{ if and only if }\text{   }\|\gamma\|_{\ell_\infty^N\otimes_\epsilon\ell_\infty^N}\leq1.
\end{align*}

In that case, if we define a correlation Bell inequality $T=(T_{i,j})_{i,j=1}^N$ as a linear functional acting on correlations, we have that, 

\begin{align}\label{Correlations NS}
LV_{\mathcal{NS}}(T)&=\frac{\sup_{\gamma\in\mathcal{NS}_c}\abs{\braket{T|\gamma}}}{\sup_{\gamma\in\mathcal{L}_c}\abs{\braket{T|\gamma}}}= \frac{\sup_{\|\gamma\|_{\ell_\infty^N\otimes_\epsilon\ell_\infty^N}\leq1}\abs{\braket{T|\gamma}}}{\sup_{\|\gamma\|_{\ell_\infty^N\otimes_\pi\ell_\infty^N}\leq1}\abs{\braket{T|\gamma}}}=\frac{\|T\|_{\ell_1^N\otimes_\pi\ell_1^N}}{\|T\|_{\ell_1^N\otimes_\epsilon\ell_1^N}}.
\end{align}

It is well known  that $\|\cdot\|_{\ell_1^N\otimes_\pi\ell_1^N}\leq \sqrt{2N} \|\cdot\|_{\ell_1^N\otimes_\epsilon\ell_1^N}$. At the same time, it is also known the existence of $u\in \ell_1^N\otimes \ell_1^N$ such that $\|u\|_{\ell_1^N\otimes_\pi\ell_1^N}\geq \sqrt{N/2}\|u\|_{\ell_1^N\otimes_\epsilon\ell_1^N}$  (see for instance \cite[Ex. 29]{LPW18} for both estimates).
Therefore, the largest non-signalling violation attainable in the correlation situation  cannot be larger than $O(\sqrt{N})$, and this order is attained. 

\section{The non-signalling norm}\label{Sec:NSnorm}

As we have seen, the relation between classical and non-signalling correlations is well understood, and this relation can be expressed and proved using the tensor norm language. We start now to follow this approach for the study of full probability distributions.

We begin by defining a suitable norm for  non-signalling probability distributions. In the following, given an element $P\in \mathbb R^{N^2K^2}=\mathbb R^{N}\otimes \mathbb R^{K}\otimes \mathbb R^{N}\otimes \mathbb R^{K}$, we will regard it as  $$P=\sum_{x,y;a,b=1}^{N,K}P(a,b,x,y)e_x\otimes e_a\otimes e_y\otimes e_b.$$

\begin{definition}\label{nsnorm}
Given $P\in \mathbb R^{N^2K^2}$, we define 
\begin{align*}
\|P\|_1=&\max_x\sum_a\max_y\sum_b\abs{P(a,b|x,y)},\\
\|P\|_2=&\max_y\sum_b\max_x\sum_a\abs{P(a,b|x,y)}.
\end{align*}

Moreover, the non-signalling norm of $P$ is defined as:
$$\|P\|_{\text{NS}}=\max\{\|P\|_1,\|P\|_2\}.$$
\end{definition}

It follows from the previous section that these three quantities are norms. Therefore, three different Banach spaces are defined: $BNS1_{NK}=(\mathbb{R}^{N^2K^2},\|\cdot\|_1)$, $BNS2_{NK}=(\mathbb{R}^{N^2K^2},\|\cdot\|_2)$ and $BNS_{NK}=(\mathbb{R}^{N^2K^2},\|\cdot\|_{\text{NS}})$. 

Recall from the introduction that, for  $P\in \R^N\otimes \R^K\otimes \R^N\otimes\R^K$, we define  $P^T=\text{flip}(P)$ where  $\text{flip}: \mathbb R^{N^2K^2}\longrightarrow \mathbb R^{N^2K^2}$ is the linear map  defined on elementary tensors by $\text{flip}(e_x\otimes e_a\otimes e_y\otimes e_b)=e_y\otimes e_b\otimes e_x\otimes e_a$. With that notation, $\|P^T\|_1=\|P\|_2$. Note that both $BNS1_{NK}$ and $BNS2_{NK}$ are isomorphic to $\ell_\infty^N(\ell_1^K(\ell_\infty^N(\ell_1^K)))$.

The following result shows that this norm indeed characterizes non-signalling distributions. 

\begin{theorem}\label{NSCB}
Let $P\in\mathbb{R}^{N^{2}K^{2}}$. Then $P\in\mathcal{NS}$ if and only if $P\in \mathcal C$  and $ \|P\|_\text{NS}=1$.
\end{theorem}

\begin{proof}
Suppose $P\in \mathcal{NS}$. Since $\mathcal{NS}\subset \mathcal C$, $P$ is belongs to $\mathcal C$. Moreover,
\begin{align*}\max_{x=1,\cdots, N}\sum_{a=1}^{K}\max_{y=1,\cdots, N}\sum_{b=1}^{K}|P(a,b|x,y)|&=\max_{x=1,\cdots, N}\sum_{a=1}^{K}\max_{y=1,\cdots, N}P_1(a|x)\\&=\max_{x=1,\cdots, N}\sum_{a=1}^{K}P_1(a|x)=1.
\end{align*}
Therefore,  $\|P\|_\text{NS}=1$.

For the other implication, suppose $P\in \mathcal C$ and  $P\notin \mathcal{NS}$. Then, we  can suppose that $P$ does not fulfill condition (\ref{ns1}) (the other case being analogous) and therefore we can assume that there exist $b_0$, $y_0$, $x_0$ and $x_1$ such that $\sum_aP(a,b_{0}|x_{0},y_{0})>\sum_aP(a,b_{0}|x_{1},y_{0})$. Then,

\begin{align*}
	&\|P\|_{\text{NS}}\geq \max_{y=1,\cdots, N}\sum _{b=1}^{K}\max_{x=1,\cdots, N}\sum_{a=1}^{K}{P(a,b|x,y)} \geq \sum _{b=1}^{K}\max_{x=1,\cdots, N}\sum_{a=1}^{K}{P(a,b|x,y_{0})}\\
	&=\sum _{b\neq b_{0}}\max_{x=1,\cdots, N}\sum_{a=1}^{K}{P(a,b|x,y_{0})}+\max_{x=1,\cdots, N}\sum_{a=1}^{K}{P(a,b_{0}|x,y_{0})}\\
	&\geq \sum _{b\neq b_{0}}\sum_{a=1}^{K}{P(a,b|x_1,y_{0})}+\sum_{a=1}^{K}{P(a,b_{0}|x_{0},y_{0})}\\
	&>\sum _{b\neq b_{0}}\sum_{a=1}^{K}{P(a,b|x_1,y_{0})}+\sum_{a=1}^{K}{P(a,b_{0}|x_{1},y_{0})}\\
	&=\sum _{b=1}^{K}\sum_{a=1}^{K}{P(a,b|x_{1},y_{0})}=1
\end{align*}

Therefore,  $\|P\|_{\text{NS}}>1$, which is a contradiction. Hence, we conclude that $P\in \mathcal{NS}$.
\end{proof}

The following set, closely related to the non-signalling probability distributions, will be very useful for our reasonings. It was introduced in \cite{LW}. 
\begin{definition}\label{d:SNOS}
The set $\text{SNOS}\subset \mathbb R^{N^2K^2}$ consists of the non-negative elements $P(a,b|x,y)$ in $\mathbb R^{N^2K^2}$ such that, for every $1\leq x, y \leq K$, there exist $(Q_{1}(a|x))_{a=1}^N$ and $(Q_{2}(b|y))_{b=1}^N$ probability distributions verifying that, for every $x,y,a,b$, $\sum_a P(a,b|x,y)\leq Q_2(b|y)$ and $\sum_b P(a,b|x,y)\leq Q_1(a|x).$
\end{definition}
\begin{remark}\label{snosproperty}
Given $P\in\mathbb{R}^{N^2K^2}$ with non-negative entries, the condition of $P$ in $\text{SNOS}$ is equivalent to the existence of $\tilde{P}$ in $\mathcal{NS}$ such that $P(a,b|x,y)\leq\tilde{P}(a,b|x,y)$ for all $x,y,a,b$ (see \cite[Claim 1]{ito} for the non-trivial implication). In this case we use  the notation $P\leq\tilde{P}$.
\end{remark}

The next result will be needed later. 
\begin{proposition}\label{p:positiveSNOS}
Let $P\in\mathbb{R}^{N^{2}K^{2}}$ have non-negative entries. Then, $P\in\text{SNOS}$  if and only if $\|P\|_{\text{NS}}\leq1.$
\end{proposition}
\begin{proof}
Let $P\in\text{SNOS}$. For every $1\leq x,y \leq K$, let $(Q_1(a|x))_{a=1}^N$, $(Q_2(b|y))_{b=1}^N$ be as in Definition \ref{d:SNOS}.  Then $$\max_x\sum_a\max_y\sum_b P(a,b|x,y)\leq\max_x\sum_a\max_yQ_1(a|x)=\max_x\sum_aQ_1(a|x)=1,$$ $$\max_y\sum_b\max_x\sum_a P(a,b|x,y)\leq\max_y\sum_b\max_xQ_2(b|y)=\max_y\sum_bQ_2(b|y)=1.$$

Consequently, $\|P\|_{\text{NS}}\leq1$.

Conversely, if $\|P\|_\text{NS}\leq1$, define for all $y$ and $b$, $\max_x\sum_aP(a,b|x,y)=\tilde{Q}_2(b|y)$ and for all $x$ and $a$, $\max_y\sum_bP(a,b|x,y)=\tilde{Q}_1(a|x)$. Thus defined,  $\tilde{Q}_1$ and $\tilde{Q}_2$ need not be probability distributions. For this reason,  we define:
\[
   Q_1(a|x)= 
   \begin{cases} 
      1-\sum_{s\neq1}\tilde{Q}_1(s|x)			 		& \mbox{if } a=1   \\
      \tilde{Q}_1(a|x)								& \mbox{if } a\neq1
   \end{cases}
   \hspace{10mm}
   Q_2(b|y)= 
   \begin{cases} 
      1-\sum_{t\neq1}\tilde{Q}_2(t|y)			 		& \mbox{if } b=1   \\
      \tilde{Q}_2(b|y)								& \mbox{if } b\neq1
   \end{cases}
\]

It is easy to see now that $Q_1, Q_2$ guarantee that  $P$ belongs to $\text{SNOS}$. 
\end{proof}

\begin{remark}
It follows from Proposition \ref{p:positiveSNOS} that the set SNOS is convex and it has the same dimension as the ambient space, $N^2K^2$
\end{remark}

Since $\|\cdot\|_\text{NS}$ is a norm in $\mathbb{R}^{N^2K^2}$, we can consider its dual norm, which is defined by 
$$\|M\|_\text{DNS}=\sup_{\|P\|_\text{NS}\leq1}|\braket{M|P}|.$$

Moreover, because $BNS_{NK}=BNS1_{NK}\cap BNS2_{NK}$, then by (\ref{sum-int}) we can say that $BNS_{NK}^*=BNS1_{NK}^*+BNS2_{NK}^*$, where $BNS1_{NK}^*=(\mathbb{R}^{N^2K^2},\|\cdot\|_1^*)$ (and similarly for $BNS2_{NK}^*$) which allows us to write  that:
\begin{align}\label{dualnorm}
\|M\|_{\text{DNS}}=\inf\{\|M_1\|_1^*+\|M_2\|_2^*:M=M_1+M_2\}.
\end{align}

Note that, according to Section \ref{Sec:Preliminaries},  for a given $M\in \mathbb R^{N^2K^2}$, we have that $$\|M\|_1^*=\sum_x\max_a\sum_y\max_b |M(a,b|x,y)|, \text{    }\text{   and }\|M\|_2^*=\sum_y\max_b\sum_x\max_a |M(a,b|x,y)|.$$ Note also that it follows from (\ref{dualnorm}) that $$\|M\|_{\text{DNS}}\leq \min\{\|M_1\|_1^*,\|M_2\|_2^*\}.$$

The next result shows how the duality works for non-signalling distributions and pointwise non-negative functionals. Theorem \ref{NS value games} follows trivially from it.
\begin{proposition}\label{non-negativeentries}
Let $M\in\mathbb{R}^{N^{2}K^{2}}$ have non-negative entries. Then,
$$\sup_{P\in\text{SNOS}}{|\braket{M|P}}|=\sup_{P\in\text{NS}}{|\braket{M|P}}|=\sup_{\|P\|_\text{NS}=1}{|\braket{M|P}}|=\|M\|_\text{DNS}.$$
\end{proposition}
\begin{proof}
Take $P\in\text{SNOS}$. Then there exists $\tilde{P}\in\text{NS}$ such that $P(a,b|x,y)\leq\tilde{P}(a,b|x,y)$. Hence,
\begin{align*}
|\braket{M|P}|&=\braket{M|P}=\sum_{a,b,x,y}{M_{x,y}^{a,b}P(a,b|x,y)}\leq\sum_{a,b,x,y}{M_{x,y}^{a,b}\tilde{P}(a,b|x,y)}\\
&\leq\sup_{P\in\text{NS}}{|\braket{M|P}}|\leq\sup_{\|P\|_\text{NS}=1}{|\braket{M|P}}|.
\end{align*}

Take now $P$ such that $\|P\|_\text{NS}=1$. Then, 
\begin{align*}
|\braket{M|P}|&=\Big|\sum_{a,b,x,y}{M_{x,y}^{a,b}P(a,b|x,y)}\Big|\leq\sum_{a,b,x,y}{|M_{x,y}^{a,b}|\cdot|P(a,b|x,y)|}\\
&=\sum_{a,b,x,y}{M_{x,y}^{a,b}\cdot|P(a,b|x,y)|}\leq{\sup_{P\in\text{SNOS}}{|\braket{M|P}}|}.
\end{align*}
Here, the last inequality follows from the fact that the element $\tilde{P}(a,b|x,y)=|P(a,b|x,y)|$ belongs to  $\text{SNOS}$ by Proposition \ref{p:positiveSNOS}, because all its coefficients are non-negative and $\|\tilde{P}\|_\text{NS}=1$.
\end{proof}

If all of the entries of a Bell inequality $M$  are non negative, Proposition \ref{non-negativeentries} shows that $\|M\|_\text{DNS}=\sup_{P\in\text{NS}}\braket{M,P}$. This has great importance when dealing with games. In a 2P1R game, $M(a,b,x,y)=\pi(x,y)V(a,b,x,y)$ where $\pi(x,y)$ is a probability distribution over the queries and $V(a,b,x,y)\in\{0,1\}$ is the function of the verifier, who outputs 1 or 0 depending on whether they have or have not won the game. The value of a game $\omega(M)$ is the probability of winning the game when using the best strategy. This value is of course different according to the resources, classical or non-signalling, that one is allowed to use. 

Since games are particular cases of Bell  inequalities with non-negative coefficients, Proposition \ref{non-negativeentries} tells us that the  non-signalling value of a game fulfills $\omega_{\mathcal{NS}}(M)=\|M\|_\text{DNS}$.

We have seen how the non-signalling value of a game (more generally, of any functional $M$ with non-negative entries) is related to  $\ell_\infty^N(\ell_1^K(\ell_\infty^N(\ell_1^K)))$ via the Banach spaces $BNS1_{NK}$ and $BNS2_{NK}$. It is also known (\cite[Section 4]{PV-Survey}) that the classical value of a game $M$ verifies $$\omega_{\mathcal L}(M) =  \|M\|_{\ell_1^N(\ell_\infty^K)\otimes_\epsilon \ell_1^N(\ell_\infty^K)}.$$

Next we  upper bound the distance between the norms associated to the classical and non-signalling probability distributions in terms of the number of inputs and outputs. These bounds will allow us to prove  Theorem \ref{thm: games}. 

First we bound this distance  in terms of the number of inputs. 

\begin{proposition}\label{cota1}
For every $M\in \mathbb R^{N^2K^2}$, 
$$\|M\|_{\ell_1^N(\ell_\infty^K(\ell_1^N(\ell_\infty^K)))}\leq N \|M\|_{\ell_1^N(\ell_\infty^K)\otimes_\epsilon\ell_1^N(\ell_\infty^K)} .$$
\end{proposition}
\begin{proof}
To prove this bound consider the following:
\begin{align*}
\|M\|_{\ell_1^N(\ell_\infty^K(\ell_1^N(\ell_\infty^K)))}&=\sum_x\max_a\sum_y\max_b\abs{M_{x,y}^{a,b}}\leq N\max_{x,a}\sum_y\max_b\abs{M_{x,y}^{a,b}}\\
&=N\|M\|_{\ell_\infty^{NK}(\ell_1^N(\ell_\infty^K))}=N\|M\|_{\ell_\infty^{NK}\otimes_\epsilon\ell_1^N(\ell_\infty^K)}\\
&=N\sup_{u\in B_{\ell_1^{NK}},v\in B_{\ell_\infty^N(\ell_1^K)}}\abs{\braket{u\otimes v|M}}\\&\leq N\sup_{u\in B_{\ell_\infty^N(\ell_1^K)},v\in B_{\ell_\infty^N(\ell_1^K)}}\abs{\braket{u\otimes v|M}}\\
&=N\|M\|_{\ell_1^N(\ell_\infty^K)\otimes_\epsilon\ell_1^N(\ell_\infty^K)},
\end{align*}
where in the third equality we have used the identification $\ell_\infty(X)=\ell_\infty^N\otimes_\epsilon X$, in the fourth equality we have used the definition of the $\epsilon$ norm (\ref{Def_eps_pi}) and in the second inequality we have used the inclusion $B_{\ell_1^{NK}}\subset B_{\ell_\infty^N(\ell_1^{K})}$.
\end{proof}

Now we bound, in the positive case, the distance between those two same norms in terms of the number of outputs. 
\begin{proposition}\label{cota4}
Let $M\in \mathbb R^{N^2K^2}$ have non-negative entries, then 
$$\|M\|_{\ell_1^N(\ell_\infty^K(\ell_1^N(\ell_\infty^K)))}\leq K \|M\|_{\ell_1^N(\ell_\infty^K)\otimes_\epsilon\ell_1^N(\ell_\infty^K)} .$$
\end{proposition}
\begin{proof}
Consider an element $M\in\mathbb{R}^{N^2K^2}$ with non-negative coefficients, then,
\begin{align*}
\|M\|_{\ell_1^N(\ell_\infty^K(\ell_1^N(\ell_\infty^K)))}&=\sum_x\max_a\sum_y\max_b\abs{M_{x,y}^{a,b}}\leq\sum_x\max_a\Big(\sum_{y,b}M_{x,y}^{a,b}\Big)\\
&=\sup_{u\in B_{\ell_\infty^N(\ell_1^K)}}\Big\langle u\Big|\sum_{y,b}M_{x,y}^{a,b}e_x\otimes e_a\Big\rangle.
\end{align*}

By the non-negativity of $M$, we may assume that $u$ is also pointwise non-negative, and we have $$\sup_{u\in B_{\ell_\infty^N(\ell_1^K)}}\Big\langle u\Big|\sum_{y,b}M_{x,y}^{a,b}e_x\otimes e_a\Big\rangle=\sup_{u\in B_{\ell_\infty^N(\ell_1^K)}, v\in B_{\ell_\infty^{NK}}}\braket{u\otimes v|M},$$
and this last supremun is attained when $v=\sum_{y,b}e_y\otimes e_b$. 

It can be checked that  $B_{\ell_\infty^{NK}}\subset K B_{\ell_\infty^N(\ell_1^K)}$. Then, we have 
\begin{align*}
\|M\|_{\ell_1^N(\ell_\infty^K(\ell_1^N(\ell_\infty^K)))}&\leq \sup_{u\in B_{\ell_\infty^N(\ell_1^K)}, v\in B_{\ell_\infty^{NK}}}\braket{u\otimes v|M}\leq \\ &\leq \sup_{v\in B_{\ell_\infty^N(\ell_1^K)},u\in B_{\ell_\infty^N(\ell_1^K)}}K\abs{\braket{u\otimes v|M}}\leq \\ &\leq K\|M\|_{\ell_1^N(\ell_\infty^K)\otimes_\epsilon\ell_1^N(\ell_\infty^K)}.
\end{align*}
\end{proof}
\begin{remark}
Note that Proposition \ref{cota1} is stated for general functionals $M$. However, Proposition \ref{cota4} requires that the element $M$ is (pointwise) non-negative. It can be seen that Proposition \ref{cota4} fails for general elements (in fact, this follows from the comments right after Equation (\ref{Correlations NS})).
\end{remark}
\begin{remark}\label{remark optimality bounds positive}
Bounds in Proposition \ref{cota1} and Proposition \ref{cota4} are sharp, as it can be checked by using the element  \[
   M(x,y,a,b)= 
   \begin{cases} 
      1		 		& \mbox{if } x=b \hspace{0.2 cm}\mbox{ and  }\hspace{0.2 cm} y=a=1,   \\
     0								& \mbox{otherwise}, 
   \end{cases}
\]where $x,y,a,b=1,\cdots, N$. Indeed, in this case we have that $N=K$, $\|M\|_{\ell_1^N(\ell_\infty^N(\ell_1^N(\ell_\infty^N)))}=N$ and $\|M\|_{\ell_1^N(\ell_\infty^N)\otimes_\epsilon\ell_1^N(\ell_\infty^N)}=1$. 
\end{remark}
\begin{proof}[Proof of Theorem \ref{thm: games}]
Since we are considering a Bell inequality with non-negative coefficients, Proposition \ref{non-negativeentries} states that  $\sup_{P\in\mathcal{NS}}\braket{M|P}=\|M\|_\text{DNS}$. At the same time, we have  $\sup_{P\in\mathcal{L}}\braket{M|P}=\|M\|_{\ell_1^N(\ell_\infty^K)\otimes_\epsilon\ell_1^N(\ell_\infty^K)}$. In addition, we have that $\|M\|_\text{DNS}\leq\min\{\|M\|_1^*,\|M\|_2^*\}$.

Now, Applying Proposition $\ref{cota1}$ and Proposition \ref{cota4}, we obtain
$$\|M\|_i^*\leq\min\{N,K\}\|M\|_{\ell_1^N(\ell_\infty^K)\otimes_\epsilon\ell_1^N(\ell_\infty^K)} \hspace{0.2 cm}(i=1,2).$$

Putting both inequalities together we get: 
\begin{align*}
LV_{\mathcal{NS}}(M)&=\frac{\sup_{P\in\mathcal{NS}}\braket{M|P}}{\sup_{P\in\mathcal{L}}\braket{M|P}}=\frac{\|M\|_\text{DNS}}{\|M\|_{\ell_1^N(\ell_\infty^K)\otimes_\epsilon\ell_1^N(\ell_\infty^K)}}\leq\frac{\min\{\|M\|_1^*,\|M\|_2^*\|\}}{\|M\|_{\ell_1^N(\ell_\infty^K)\otimes_\epsilon\ell_1^N(\ell_\infty^K)}}\\&\leq\min\{N,K\}.
\end{align*}
\end{proof}

\begin{remark}
Obvious modifications of these proofs show that if we distinguish the inputs and outputs for Alice and Bob as $N_1$, $N_2$, $K_1$ and $K_2$, then one has the following bound for pointwise non-negative elements $M$: 
$$LV_{\mathcal{NS}}(M)\leq\min\{N_1,N_2,K_1,K_2\}.$$
\end{remark}
\begin{remark}
One could wonder whether the element in Remark \ref{remark optimality bounds positive} can be used to give an optimal ratio $LV_{\mathcal{NS}}(M)=\|M\|_\text{DNS}/\|M\|_{\ell_1^N(\ell_\infty^K)\otimes_\epsilon\ell_1^N(\ell_\infty^K)}$. However, for that element $M$ it is easy to see that $\|M\|_2^*=\|\text{flip}(M)\|_{\ell_1^N(\ell_\infty^N(\ell_1^N(\ell_\infty^N)))}=1$. Hence, $\|M\|_{\text{DNS}}\leq 1$ and the ratio in this case would not be greater than one.
\end{remark}
\section{Optimal lower bounds}\label{Sec:Game}

In this section we show that the upper bounds obtained in the previous section are essentially optimal. To do so, we consider a random family of games and show that, with high probability, the games in our family attain the  upper bounds in Section \ref{Sec:NSnorm}, up to a logarithmic factor.

Consider a family of elements $\{\sigma_{xy}\}_{x,y=1}^{N}$ where $\sigma_{xy}$ is in $S_K$, the symmetric group over $[N]$; that is, the group of permutations of the inputs. For every such family we  define the linear functional with non negative entries:
$$M=\sum_{x,y=1}^{N}\sum_{j=1}^{K} e_x\otimes e_j\otimes e_y\otimes e_{\sigma_{xy}(j)}.$$

For the interested reader, we remark that, properly normalized, $M$  can be seen as a unique game \cite{FL92}, with the uniform distribution on the inputs $(x,y)$ and the  verifier function defined as 1 if and only if $b=\sigma_{xy}(a)$, and 0 otherwise. We will not explicitly use this fact, though. 

\

\textbf{NS value of $M$:}

We prove next that $\|M\|_{\text{DNS}}=N^2$. We consider the following strategy:
$$P=\frac{1}{K}\sum_{x,y=1}^N\sum_{j=1}^K e_x\otimes e_j\otimes e_y\otimes e_{\sigma_{xy}(j)}.$$ 

It can be seen that it is a non-signalling probability distribution because $\|P\|_{\text{NS}}=1$ and all of its entries are positive. Then if we consider the value of  $M$ acting on $P$ we obtain:

\begin{align*}
\braket{M|P}&=\Big\langle \sum_{x,y}\sum_j e_x\otimes e_j\otimes e_y\otimes e_{\sigma_{xy}(j)}\Big|\frac{1}{K}\sum_{x',y'}\sum_{j'} e_{x'}\otimes e_{j'}\otimes e_{y'}\otimes e_{\sigma_{x'y'}(j')}\Big\rangle\\
&=\sum_{x,x',y,y',j,j'}\frac{1}{K}\braket{e_x|e_{x'}}\braket{e_j|e_{j'}}\braket{e_y|e_{y'}}\braket{e_{\sigma_{xy}(j)}|e_{\sigma_{x'y'}(j')}}=\sum_{x,y,j}\frac{1}{K}=N^2.
\end{align*}

Therefore we have that  $\|M\|_{\text{DNS}}\geq N^2$. At the same time, it is easy to see that $\langle M | P\rangle \leq N^2$ for every  $P\in \mathcal C$ (even a signalling one). Hence $\|M\|_\text{DNS}=N^2$.

\

\textbf{Classical value of $M$:}

We study now the classical value of $M$. As $\mathcal{L}$ is a convex polytope \cite{Tsirelson} and $M$ is a convex (in fact, linear) function acting on $\mathcal{L}$, applying convexity arguments it is clear that we only need to consider classical extremal strategies. A classical extremal strategy $P$ is uniquely determined by two functions $a,b:\{1,\ldots,N\}\longrightarrow\{1,\ldots,K\}$ in such a way that:
$$P=\sum_{x,y}e_x\otimes e_{a(x)}\otimes e_y\otimes e_{b(y)}.$$

Then, $M$ acting on $P$ verifies:
\begin{align*}
\braket{M|P}&=\Big\langle\sum_{x,y}\sum_j e_x\otimes e_j\otimes e_y\otimes e_{\sigma_{xy}(j)}\Big|\sum_{x',y'}e_{x'}\otimes e_{a(x')}\otimes e_{y'}\otimes e_{b(y')}\Big\rangle\\
&=\sum_{x,x',y,y',j}\braket{e_x|e_{x'}}\braket{e_j|e_{a(x')}}\braket{e_y|e_{y'}}\braket{e_{\sigma_{xy}(j)}|e_{b(y')}}=\sum_{x,y}\braket{e_j|e_{a(x)}}\braket{e_{\sigma_{xy}(j)}|e_{b(y)}}\\
&=\sum_{x,y}\braket{e_{\sigma_{xy}(a(x))}|e_{b(y)}}.
\end{align*}

We apply now probabilistic reasonings. For every $1\leq x,y\leq N$, we consider the  permutation $\sigma_{xy}$ to be a random variable uniformly distributed in $S_K$. For $(x,y)\not = (x',y')$ we consider $\sigma_{xy}$ and $\sigma_{x'y'}$ to be independent random variables. That is, $M$ is a random variable in the probability space $\Omega:= (S_K)^{\otimes N^2}$, considered with the uniform probability.

We fix a  classical extremal strategy $P$ characterized by  functions $a,b:\{1,\ldots,N\}\longrightarrow\{1,\ldots,K\}$ as above. For one such $P$ and for every pair of inputs $x$ and $y$, we can define a random variable $Z^P_{x,y}:S_K\rightarrow \{0,1\}$ by 
\begin{align*}
  Z_{x,y}^P= \braket{e_{\sigma_{x,y}(a(x))}|e_{b(y)}}.
\end{align*}

Recall that the superindex $P$ makes reference to the extremal probability distribution, which uniquely determines the functions $a$ and $b$. This random variable takes the following values:
\[ Z_{x,y}^P=\begin{cases} 
      1 & \text{if } \sigma_{x,y}(a(x))=b(y),\\
      0 & \text{if } \sigma_{x,y}(a(x))\neq b(y).
   \end{cases}
\]

Clearly,  the probability of $\sigma_{x,y}(a(x))=b(y)$ is $1/K$. Therefore $Z_{x,y}^P$ is a Bernoulli  variable of parameter $1/K$.

We recall the  following Chernoff-type bound \cite{chernoff}: 
\begin{theorem}\label{t:Chernoff}
Let $X_1$, $X_2$, $\ldots$, $X_n$ be independent $0$-$1$ random variables with $\mathbb{P}\left[X_i=1\right]=p_i$. Denote $X=\sum_{i=1}^nX_i$ and $\mu=\mathbb{E}[X]$. Then for all $\delta>1$:
$$\mathbb{P}\left[X\geq(1+\delta)\mu\right]\leq\exp\Big(\frac{-\delta^2\mu}{2+\delta}\Big).$$
\end{theorem}

We define a  new random variable  $$Z^P=\sum_{x,y=1}^N Z_{x,y}^P= \sum_{x,y=1}^N\braket{e_{\sigma_{xy}(a(x))}|e_{b(y)}}.$$

Clearly, if $(x,y)\neq (x',y')$, then $Z_{x,y}^P$ and $Z_{x',y'}^P$ are independent. It is easy to see that $\mathbb{E}\left[Z^P\right]=N^2/K$; we can choose then  $N=K$ and apply Theorem \ref{t:Chernoff} to obtain:
$$\mathbb{P}\left[Z^P\geq(1+\delta)N\right]\leq\exp\Big(\frac{-\delta^2 N}{2+\delta}\Big).$$

There are $N^N$ different possibilities for the $a$ function and also for the $b$ function. That means, there are in total $N^{2N}$ different classical extremal strategies, which we  label as $P_i$ for $i=1,\ldots,N^{2N}$.

Now we apply the union bound and we obtain
\begin{align*}
\mathbb{P}\left[ \bigcup_{i=1}^{N^{2N}}\left(Z^{P_i}\geq(1+\delta)N\right)\right]&\leq\sum_{i=1}^{N^{2N}}\mathbb{P}\left[Z^{P_i}\geq(1+\delta)N\right]\leq N^{2N}\exp\Big(\frac{-\delta^2 N}{2+\delta}\Big)\\
&=\exp\Big(\log N^{2N}\Big)\exp\Big(\frac{-\delta^2 N}{2+\delta}\Big)=\exp\Big(2N\log N-\frac{\delta^2 N}{2+\delta}\Big)
\end{align*}

Choosing  $\delta=3\log N-2$, we have
\begin{align*}
\exp\Big(2N\log N-\frac{\delta^2 N}{2+\delta}\Big)&=\exp\Big(2N\log N-\frac{(3\log N)^2-12\log N+4}{3\log N}N\Big)\\
&=\exp\Big(-N(\log N-4)-\frac{4N}{3\log N}\Big)<1,
\end{align*}
for $N\geq 5$. 
Therefore
$$\mathbb{P}\left[\Bigg(\bigcup_{i=1}^{N^{2N}}\left(Z^{P_i}\geq(3\log N-1)N\right)\Bigg)^c\right]=\mathbb{P}\left[\bigcap_{i=1}^{N^{2N}}\left(Z^{P_i}<(3\log N-1)N\right)\right]>0,$$ 
for $N\geq 5$.

Hence, we know the existence of a family of  $N^2$ permutations, $(\sigma_{x,y})_{x,y=1}^N$ defining a linear functional $M$, such that $\|M\|_\text{DNS}=N^2$ and $\|M\|_{\ell_1^N(\ell_\infty^N)\otimes_\epsilon\ell_1^N(\ell_\infty^N)}\leq (3\log N-1)N$. This concludes the analysis of the classical bound.
\begin{remark}
The same result can be obtained with a  more restrictive type of games, the XOR-d games considered in \cite{BS15}. Alice and Bob receive questions $(x,y)$ from $X\times Y$ and reply with answers $a,b\in (\mathbb{Z}_N,+)$ where $\mathbb{Z}_N$ is the cyclic group of $N$ elements with inner operation $+$. The winning constraint is now $a+b=\sigma_{x,y}$ for some function $\sigma:X\times Y\rightarrow \mathbb{Z}_N$, $\sigma(x,y)=\sigma_{x,y}$. Choosing  $\sigma_{xy}$ uniformly and independently (in $(x,y)$)  we obtain the same bounds.  
\end{remark}
\section{A tensor norm description of the non-signalling set}\label{Sec: General Bell}

As we have seen, we can embed $\mathcal{NS}$ into $\mathbb{R}^{N^2K^2}$ and consider the non-signalling norm in this space. We have already explained that this  procedure suits perfect to relate norms and values (local, non-signalling) of  Bell inequalities with non-negative coefficients. But if we consider  general Bell functionals, with coefficients not necessarily non-negative, then the relation between  $\omega_{NS}(M)$ and  $\|M\|_\text{DNS}$ is not so clear anymore (it can be easily checked that $\omega_{NS}(\cdot)$ is not a norm in $\mathbb R^{N^2K^2}$). 

In order to understand this situation, we follow an approach similar to what was done in \cite{JungeP11low}. In that paper, in order to study quantum violation of general Bell inequalities, the authors introduced an auxiliary Banach space $NSG(N,K)$ defined as the linear space
\begin{align*}
NSG(N,K)=\{\{R(x|a)\}_{x,a=1}^{N,K}\in \R^{NK}: \sum_{a=1}^KR(x|a)={\rm constant}\in \R  \text{   for every   } x\},
\end{align*}endowed with the norm 
\begin{align*}
||R||_{NSG(N,K)}=\inf \{|\lambda|+|\mu|:R=\lambda P+\mu Q: P,Q\in S(N,K)\},
\end{align*} where
\begin{align*}
S(N,K)=\{\{P(x|a)\}_{x,a=1}^{N,K}: P(x|a)\geq 0 \text{   for every   } x,a   \text{   and   } \sum_{a=1}^KP(x|a)=1  \text{   for every   } x\}.
\end{align*}

The following result was proved in \cite{JungeP11low}.
\begin{theorem}\label{Thm JP}
The following relations hold: $co(\mathcal{L}\cup-\mathcal{L})=B_{NSG\otimes_\pi NSG}$ and $NSG\otimes_\pi NSG$ is isomorphic to $(\ell_{\infty}^N(\ell_1^{K-1})\oplus_\infty\mathbb{R})\otimes_\pi(\ell_{\infty}^N(\ell_1^{K-1})\oplus_\infty\mathbb{R})$. Moreover, the Banach-Mazur distance between these two space is less or equal than 9, independently of the dimension.
\end{theorem}

Here, proving that $co(\mathcal{L}\cup-\mathcal{L})=B_{NSG\otimes_\pi NSG}$ is very easy from the definition of $NSG$ and the norm $||\cdot||_{NSG(N,K)}$ (see \cite[Lemma 14, part a)]{JungeP11low}). Instead of using the rest of the estimates in the previous theorem, we will prove the following lemma, which will be enough for our purpose in this work.

\begin{lemma}\label{NSG-L1}
Let $R\in NSG(N,K)$. Then, $$||R||_{NSG(N,K)}=\max_{x=1,\cdots, N}\sum_{a=1}^K|R(x|a)|.$$
\end{lemma}
\begin{proof}
Since it is clear that for every $P\in S(N,K)$, we have $\max_{x=1,\cdots, N}\sum_{a=1}^K|P(x|a)|=1$, then it is obvious, by triangle inequality, that $\max_{x=1,\cdots, N}\sum_{a=1}^K|R(x|a)|\leq ||R||_{NSG(N,K)}$ for every $R\in NSG(N,K)$.

In order to show the converse inequality, let us consider an element in $R\in NSG(N,K)$ such that $\max_{x=1,\cdots, N}\sum_{a=1}^K|R(x|a)|\leq 1$ and we will show that $||R||_{NSG(N,K)}\leq 1$. The proof for general elements follows trivially by re-normalizing them.

Let us denote, for a fixed $x$, $$A_x^+=\{a:R(x|a)\geq 0\}\hspace{0.3 cm}\text{  and   }\hspace{0.3 cm}A_x^-=\{a:R(x|a)< 0\},$$and$$M=\max_x\sum_{a\in A_x^+}R(x|a)\hspace{0.3 cm}\text{  and   }\hspace{0.3 cm}m=\max_x\Big|\sum_{a\in A_x^-}R(x|a)\Big|.$$

The fact that $\sum_aR(x|a)=K$ for every $x$ guarantees that the previous $\max$ and $\min$ are attained in the same $x$. In particular, note that $M-m=K$ and $M+m=\max_x\sum_a|R(x|a)|.$
Therefore, we can write $R=MP_1- mP_2$, where we define for each $x$:
\begin{align*}
P_1(x|a)=\begin{cases}
      \frac{R(x|a)}{M} & \text{if } R(x|a)\geq 0 \text{ and } 1\leq a\leq K-1 \\
      0 & \text{if } R(x|a)< 0 \text{ and }  1\leq a\leq K-1\\
      1-\sum_{a=1}^{K-1}P_1(x|a) & \text{if } a=K
   \end{cases}
\end{align*}
\begin{align*}
P_2(x|a)=\begin{cases}
      -\frac{R(x|a)}{m} & \text{if } R(x|a)< 0 \text{ and } 1\leq a\leq K-1 \\
      0 & \text{if } R(x|a)\geq 0 \text{ and }  1\leq a\leq K-1\\
      1-\sum_{a=1}^{K-1}P_2(x|a) & \text{if } a=K
   \end{cases}
\end{align*}
Since $P_1,P_2\in S(N,K)$ we conclude that $\|R\|_{NSG}\leq 1$ and we finish the proof.
\end{proof}
\begin{remark}\label{Norm-complementation}
Using the notation from Section \ref{Sec:Preliminaries}, the previous lemma says that the identity map $id:NSG\rightarrow \ell_\infty^N(\ell_1^K)$ is an isometry. However, the fact that the projective tensor norm is not injective means that $id:NSG\otimes_\pi NSG\rightarrow \ell_\infty^N(\ell_1^K)\otimes_\pi \ell_\infty^N(\ell_1^K)$ does not need to be an isometry anymore. This is the main reason to introduce the space $(\ell_{\infty}^N(\ell_1^{K-1})\oplus_\infty\mathbb{R})$ in Theorem \ref{Thm JP}.
\end{remark}

In the following we make a similar construction of a normed space based upon the non-signalling distributions. This space, called $\mathcal{ANS}$, has the property that its elements fulfill conditions (\ref{ns1}) and (\ref{ns2}) and also that $B_{\mathcal{ANS}}=co(\mathcal{NS}\cup-\mathcal{NS})$, as it will be shown later. At the end of this section we use both of these spaces (the classical and the non-signalling), to prove Theorem \ref{boundlv}. This space was already defined at the introduction, but for the convenience of the reader we recall here its definition.

\begin{definition}\label{BNS}
Let $\mathcal{ANS}$ consist of the elements $R\in \R^{N^2K^2}$ for which there exist $\{Q(y,b)\}_{y,b}\in\R^{NK}$, $\{P(x,a)\}_{x,a}\in\R^{NK}$ and a constant $z\in \mathbb R$ such that $\sum_aR(x,y,a,b)=Q(y,b)$ for all $x,b,y$, $\sum_bR(x,y,a,b)=P(x,a)$ for all $y,b,x$ and $\sum_{a,b}R(x,y,a,b)=z$ for all $x, y$. 

We consider $\mathcal{ANS}$ endowed with the restriction to it of the  non-signalling norm. 
\end{definition}

We will need some notation. Given $R=(R(x,y,a,b))_{x,y=1, a,b=1}^{N,K}\in \mathbb R^{N^2K^2}$ we define

$$R^+(x,y,a,b)=\begin{cases} 
      R(x,y,a,b) & \text{if }R(x,y,a,b)\geq 0,\\
     0 & \text{otherwise. }
   \end{cases}$$
$$R^-(x,y,a,b)=\begin{cases}  
      R(x,y,a,b) & \text{if }R(x,y,a,b)< 0, \\
     0 & \text{otherwise. }
   \end{cases}$$
Clearly  $R=R^++R^-$

We will use the following notation for fixed $x,a$ and $y,b$ respectively:
\begin{align*} 
     c_{xa}&=\max_y\sum_b\abs{R(x,y,a,b)}=\sum_b\abs{R(x,y_{xa},a,b)}, \\
     c_{xa}^{\pm}&=\sum_b\abs{R^{\pm}(x,y_{xa},a,b)},\\
     d_{yb}&=\max_x\sum_a\abs{R(x,y,a,b)}=\sum_b\abs{R(x_{yb},y,a,b)},\\
     d_{yb}^{\pm}&=\sum_b\abs{R^{\pm}(x_{yb},y,a,b)}.
\end{align*}

It is straightforward to check that for every $x,y,a,b$ one has the following equalities:
\begin{align*}
c_{xa}&=c_{xa}^++c_{xa}^-,\\
P(x,a)&=c_{xa}^+-c_{xa}^-,\\
d_{yb}&=d_{yb}^++d_{yb},\\
Q(y,b)&=d_{yb}^+-d_{yb}^-,\\
z&=\sum_a c_{x a}^+-\sum_a c_{xa}^-=\sum_b d_{yb}^+-\sum_b d_{yb}^-.
\end{align*}

We will need the following two lemmas. 
\begin{lemma}\label{lemma2}
If $R\in \mathcal{ANS}$, then 
$$\|R\|_{\text{NS}}=\|R^+\|_{\text{NS}}+\|R^-\|_{\text{NS}}.$$

Moreover, if $\|R\|_{\text{NS}}=\sum_{a,b}\abs{R(x_0,y_a,a,b)}$, then $\|R^+\|_{\text{NS}}=\sum_{a,b}|R^+(x_0,y_a,a,b)|$ and $\|R^-\|_{\text{NS}}=\sum_{a,b}\abs{R^-(x_0,y_a,a,b)}$. An analogous statement holds if $\|R\|_{\text{NS}}=\sum_{a,b}\abs{R(x_b,y_0,a,b)}$.
\end{lemma}
\begin{proof}
Consider an element $R\in\mathcal{ANS}$ from Definition \ref{BNS} with its notation. Suppose that $\|R\|_1=\sum_{a,b}R(x_0,y_a,a,b)$, $\|R\|_2=\sum_{a,b}R(x_b,y_0,a,b)$ and also assume, without loss of generality, that $\|R\|_1\geq\|R\|_2$. Using the notation introduce above, we have in addition, 
\begin{align*}
\|R\|_{\text{NS}}&=\sum_a c_{x_0 a}^++\sum_a c_{x_0a}^-.
\end{align*}

We are going to divide the proof in three steps. First, we note  that $$\max_y\sum_b\abs{R^+(x,y,a,b)}=c_{xa}^+ \hspace{0.3 cm}\text{    and   }\hspace{0.3 cm}\max_y\sum_b\abs{ R^-(x,y,a,b)}=c_{xa}^-.$$

To see this, recall that it follows from adding or subtracting the next equality and inequality, which hold for every $1\leq y \leq N$:
$$
\begin{rcases}
c_{xa}^+-c_{xa}^-=P(x,a)=\sum_b \abs{R^+(x,y,a,b)}-\sum_b \abs{R^-(x,y,a,b)} \\
c_{xa}^++c_{xa}^-=c_{xa}\geq\sum_b \abs{R^+(x,y,a,b)}+\sum_b \abs{R^-(x,y,a,b)}
\end{rcases}
$$

Similarly, the same result is obtained for $d_{yb}^+$ and $d_{yb}^-$.

In the second step we prove that $$\max_x\sum_a c_{xa}^+=\sum_a c_{x_0a}^+ \hspace{0.3 cm}\text{    and   }\hspace{0.3 cm}\max_x\sum_a c_{xa}^-=\sum_a c_{x_0a}^-.$$ Again, this follows from adding and subtracting the next equality and inequality, both of which clearly hold $1\leq x \leq N$. 
$$
\begin{rcases}
\sum_a c_{x_0a}^+-\sum_a c_{x_0a}^-=z=\sum_{a} c_{xa}^+-\sum_{a} c_{xa}^- \\
\sum_a c_{x_0a}^++\sum_a c_{x_0a}^-=\|R\|_1\geq \sum_{a} c_{xa}^++\sum_{a} c_{xa}^-
\end{rcases}
$$

Similarly one proves that  $\max_y\sum_b d_{yb}^+=\sum_b d_{y_0b}$ and $\max_y\sum_b d_{yb}^-=\sum_b d_{y_0b}$ using $\|R\|_2$ instead of $\|R\|_1$.

The third step consists on showing that actually $$\|R^+\|_{\text{NS}}=\sum_a c_{x_0 a}^+\hspace{0.3 cm}\text{    and   }\hspace{0.3 cm}\|R^-\|_{\text{NS}}=\sum_a c_{x_0 a}^-.$$ To do this, note that the next equalities and inequalities clearly hold for every $1\leq y \leq N$:
$$
\begin{rcases}
\sum_a c_{x_0a}^+-\sum_a c_{x_0a}^-=z=\sum_{b} d_{yb}^+-\sum_{b} d_{yb}^- \\
\sum_a c_{x_0a}^++\sum_a c_{x_0a}^-=\|R\|_\text{NS}=\|R\|_1\geq\|R\|_2= \sum_{a} d_{yb}^++\sum_{a} d_{yb}^-
\end{rcases}
$$

This shows that  $\sum_a c_{x_0a}^+\geq\sum_{b} d_{yb}^+$ and $\sum_a c_{x_0a}^-\geq\sum_{b} d_{yb}^-$ for all $y$, which finishes the proof. 
\end{proof}
\begin{remark}\label{remarkdif}
Using Lemma \ref{lemma2} and its notation, if $R\in \mathcal{ANS}$, then it follows that $$\|R^+\|_\text{NS}-\|R^-\|_\text{NS}=\sum_{ab}|R^+(x_0,y_a,a,b)|-\sum_{ab}\abs{R^-(x_0,y_a,a,b)}=\sum_{ab}R(x_0,y_a,a,b)=z.$$
\end{remark}

The following lemma is an adapted version of \cite[Claim 1]{ito}. The proof is analogous and for completeness it will be given in full detail. Recall that given a set $\mathcal{A}\subset\mathbb{R}^M$ with $M\in\mathbb{N}$ we can define $r\mathcal{A}=\{ra\hspace{0.2 cm}\text{such that} \hspace{0.2 cm} a\in\mathcal{A}\}$ for $r\in\mathbb{R}^+$.

\begin{lemma}\label{lemmaito}
Given $P=(P(a,b,x,y))_{a,b,x,y}\in \mathbb R^{N^2K^2}$ with non-negative entries, suppose that there exist $(Q_1(x,a))_{x,a}$ and $(Q_2(y,b))_{y,b}$ such that $\sum_{a}P(a,b,x,y)\leq Q_2(y,b)$ for all $x,y,b$, $\sum_bP(a,b,x,y)\leq Q_1(x,a)$ for all $x,y,a$ and $\sum_a Q_1(x,a)=\sum_b Q_2(y,b)=\|P\|_\text{NS}$ for all $x,y$, then there exists $\tilde{P}\in \|P\|_\text{NS}\mathcal{NS}$ such that $P(a,b,x,y)\leq\tilde{P}(a,b,x,y)$ for all $x,y,a,b$.
\end{lemma}

\begin{proof}
Defining $u_{xy}=\|P\|_\text{NS}-\sum_{a,b}P(a,b,x,y)$, $t_{xy}(b)=Q_2(y,b)-\sum_a P(a,b,x,y)$ and $s_{xy}(a)=Q_1(x,a)-\sum_bP(a,b,x,y)$ we can construct $\tilde{P}\in\|P\|_\text{NS}\mathcal{NS}$ using:

$$
\tilde{P}(a,b,x,y)=\begin{cases} 
     P(a,b,x,y)+\frac{s_{xy}(a)t_{xy}(b)}{u_{xy}} & \text{if }u_{xy}>0, \\
     P(a,b,x,y) & \text{if }u_{xy}=0.
\end{cases}
$$

To show that $\sum_{a}\tilde{P}(x,y,a,b)=Q_2(y,b)$, consider first the case $u_{xy}\neq0$:

\begin{align*}
\sum_a \tilde{P}(a,b,x,y)&=\sum_a P(a,b,x,y)+\frac{(\sum_a s_{xy}(a))t_{xy}(b)}{u_{xy}}\\
&=\sum_a P(a,b,x,y)+\frac{(\|P\|_\text{NS}-\sum_{ab}P(x,y,a,b))t_{xy}(b)}{\|P\|_\text{NS}-\sum_{ab}P(x,y,a,b)}\\
&=\sum_a P(a,b,x,y)+t_{xy}(b)\\
&=\sum_a P(a,b,x,y)+Q_2(y,b)-\sum_a P(a,b,x,y)=Q_2 (y,b).
\end{align*}

On the other side, the case $u_{xy}=0$ (which implies $\sum_{ab}P(a,b,x,y)=\|P\|_\text{NS}$) is incompatible with having $\sum_a P(a,b,x,y)<Q_2(y,b)$, because this last implies $\sum_{ab}P(a,b,x,y)<\sum_b Q_2(y,b)=\|P\|_\text{NS}$. Hence, in this case we also have $\sum_{a}\tilde{P}(x,y,a,b)=Q_2(y,b)$.

It can be seen analogously that $\sum_b \tilde{P}(a,b,x,y)=Q_1(x,a)$. Moreover $\tilde{P}$ has the property $\sum_{ab}\tilde{P}(a,b,x,y)=\sum_b Q_2(y,b)=\|P\|_\text{NS}$ for all $x,y$.
\end{proof}

\begin{proof}[Proof of Theorem \ref{NS description}]
Take $R\in B_{\mathcal{ANS}}$. We are aiming to obtain $\tilde{R}^+$ from $R^+$ and $\tilde{R}^-$ from $R^-$ in such a way that $\tilde{R}^\pm\in\|R^\pm\|_\text{NS}\mathcal{NS}$  and $R^++R^-=\tilde{R}^+-\tilde{R}^-$.

In that case we would have:
$$R=R^++R^-=\tilde{R}^+-\tilde{R}^-=\|R\|_\text{NS}\cdot\frac{\tilde{R}^+}{\|R^+\|_\text{NS}}-\|R^-\|_\text{NS}\cdot\frac{\tilde{R}^-}{\|R^-\|_\text{NS}}$$

Since $\tilde{R}^+/\|R^+\|_\text{NS}\in\mathcal{NS}$ and $\tilde{R}^-/\|R^-\|_{\text{NS}}\in\mathcal{NS}$, and also by Lemma \ref{lemma2}, $\|R^+\|_\text{NS}+\|R^-\|_\text{NS}=\|R\|_\text{NS}\leq1$, we would conclude that $R\in co(\mathcal{NS}\cup -\mathcal{NS})$.

Using the definitions of $c_{xa}^\pm$ and $d_{yb}^\pm$ from Lemma \ref{lemma2}, let for all $x$:

$$
Q_1^{\pm}(x,a)=\begin{cases} 
     c_{xa}^{\pm} & \text{if }a=1,\ldots,K-1, \\
     \|R^\pm\|_\text{NS}-\sum_{a=1}^{K-1}c_{xa}^\pm & \text{if }a=K.
\end{cases}
$$

$$
Q_2^{\pm}(y,b)=\begin{cases} 
     d_{yb}^{\pm} & \text{if }b=1,\ldots,K-1, \\
     \|R^\pm\|_\text{NS}-\sum_{a=1}^{K-1}d_{yb}^\pm & \text{if }b=K.
\end{cases}
$$

Let us show that $Q_1^\pm$ and $Q_2^\pm$ fulfill conditions of Lemma \ref{lemmaito}. For $Q_1^+$ the justification is the following (and for the rest it can be proven similarly): On the one hand, it is clear that $Q_1^{\pm}(x,a)\geq 1$ for every $x,a$, and $$\sum_{a=1}^KQ_1^{\pm}(x,a)=\|R^\pm\|_\text{NS}.$$

On the other hand, for a fixed $x$,
\begin{align*}
\sum_b R^+(x,y,a,b)&\leq \sup_y\sum_b R^+(x,y,a,b)=c_{xa}^+=Q_1^+(x,a) \hspace{0.3 cm} \text{for all } \hspace{0.3 cm} a=1,\cdots, K-1,\\
\sum_b R^+(x,y,K,b)&\leq c_{xK}^+\leq \|R^+\|_\text{NS}-\sum_{a=1}^{K-1} c_{xa}^+=Q_1^+(x,K).
\end{align*}

Since we can see analogously that $\sum_a R^+(x,y,a,b)\leq Q_2^+(y,b)$ for every $x, y,b$, we can apply Lemma \ref{lemmaito} to $R^+$ using $Q_1^+$ and $Q_2^+$ to obtain $\tilde{R}^+\in\|R^+\|_\text{NS}\mathcal{NS}$ and such that $R^+(x,y,a,b)\leq \tilde{R}^+(x,y,a,b)$ for every $x,y,a,b$. Moreover, we can show analogously that Lemma \ref{lemmaito} can be applied to $|R^-|$ using $Q_1^-$ and $Q_2^-$ to obtain $\tilde{R}^-\in\|R^-\|_\text{NS}\mathcal{NS}$ and such that $-R^-(x,y,a,b)=|R^-(x,y,a,b)|\leq \tilde{R}^-(x,y,a,b)$ for every $x,y,a,b$\footnote{Note that Lemma \ref{lemmaito} applies on non-negative tensors, so we must use it on $-R^-=|R^{-}|$.}.

We still have to prove that $R=\tilde{R}^+-\tilde{R}^-$. Note that in the construction of $\tilde{R}^+$ and $\tilde{R}^-$ using Lemma \ref{lemmaito} one defines:

\begin{align*}
s_{xy}^{\pm}(a)&=c_{xa}^{\pm}-\sum_{b=1}^{K}|R^{\pm}(x,y,a,b)|\text{ for $a=1,\ldots,K-1$},\\
s_{xy}^{\pm}(K)&=\|R^{\pm}\|_\text{NS}-\sum_{a=1}^{K-1}c_{xa}^{\pm}-\sum_{b=1}^{K}|R^{\pm}(x,y,K,b)|,\\
t_{xy}^{\pm}(b)&=d_{yb}^{\pm}-\sum_{a=1}^{K}|R^{\pm}(x,y,a,b)|\text{ for $b=1,\ldots,K-1$},\\
t_{xy}^{\pm}(K)&=\|R^{\pm}\|_\text{NS}-\sum_{b=1}^{K-1}d_{yb}^{\pm}-\sum_{a=1}^{K}|R^{\pm(x,y,a,K)|},\\
u_{xy}^{\pm}&=\|R^{\pm}\|_\text{NS}-\sum_{a,b=1}^K |R^{\pm}(x,y,a,b)|.
\end{align*}

\iffalse
\begin{align*}
s_{xy}^+(a)&=c_{xa}^+-\sum_{b=1}^{K}R^+(x,y,a,b)\text{ for $a=1,\ldots,K-1$}\\
s_{xy}^+(K)&=\|R^+\|_\text{NS}-\sum_{a=1}^{K-1}c_{xa}^+-\sum_{b=1}^{K}R^+(x,y,K,b)\\
t_{xy}^+(b)&=d_{yb}^+-\sum_{a=1}^{K}R^+(x,y,a,b)\text{ for $b=1,\ldots,K-1$}\\
t_{xy}^+(K)&=\|R^+\|_\text{NS}-\sum_{b=1}^{K-1}d_{yb}^+-\sum_{a=1}^{K}R^+(x,y,a,K)\\
u_{xy}^+&=\|R^+\|_\text{NS}-\sum_{a,b=1}^K R^+(x,y,a,b)
\end{align*}
\begin{align*}
s_{xy}^-(a)&=c_{xa}^--\sum_{b=1}^{K}\abs{R^-(x,y,a,b)}\text{ for $a=1,\ldots,K-1$}\\
s_{xy}^-(K)&=\|R^-\|_\text{NS}-\sum_{a=1}^{K-1}c_{xa}^--\sum_{b=1}^{K}\abs{R^-(x,y,K,b)}\\
t_{xy}^-(b)&=d_{yb}^--\sum_{a=1}^{K}\abs{R^-(x,y,a,b)}\text{ for $b=1,\ldots,K-1$}\\
t_{xy}^-(K)&=\|R^-\|_\text{NS}-\sum_{b=1}^{K-1}d_{yb}^--\sum_{a=1}^{K}\abs{R^-(x,y,a,K)}\\
u_{xy}^-&=\|R^-\|_\text{NS}-\sum_{a,b=1}^K\abs{R^-(x,y,a,b)}
\end{align*}
\fi

In order to obtain: 

$$\tilde{R}^+(x,y,a,b)=\begin{cases} 
     R^+(x,y,a,b)+\frac{s_{xy}^+(a)t_{xy}^+(b)}{u_{xy}^+} & \text{if }u_{xy}^+>0 \\
     R^+(x,y,a,b) & \text{if }u_{xy}^+=0
     \end{cases}
     $$
     
$$
\tilde{R}^-(x,y,a,b)=\begin{cases} 
     |R^-(x,y,a,b)|+\frac{s_{xy}^-(a)t_{xy}^-(b)}{u_{xy}^-} & \text{if }u_{xy}^->0 \\
     |R^-(x,y,a,b)| & \text{if }u_{xy}^-=0
\end{cases}
$$

In order to show $R=\tilde{R}^+-\tilde{R}^-$ we will prove that $s_{xy}^+(a)=s_{xy}^-(a)$, $t_{xy}^+(b)=t_{xy}^-(b)$ and $u_{xy}^+=u_{xy}^-$ for all $x,y,a,b$, from where the result follows straightforwardly.

On the one hand, Remark \ref{remarkdif} guarantees that
\begin{align*}
&u_{xy}^+=u_{xy}^-
\Leftrightarrow \|R^+\|_\text{NS}-\sum_{a,b=1}^K|R^+(x,y,a,b)|=\|R^-\|_\text{NS}-\sum_{a,b=1}^K \abs{R^-(x,y,a,b)}\\
&\Leftrightarrow \|R^+\|_\text{NS}-\|R^-\|_\text{NS}=\sum_{a,b=1}^K R^+(x,y,a,b)+\sum_{a,b=1}^K\abs{R^-(x,y,a,b)}=\sum_{a,b=1}^K R(x,y,a,b)=z.
\end{align*}
 
On the other hand, for all $a=1,\cdots, K-1$,
\begin{align*}
&s_{xy}^+(a)=s_{xy}^-(a)\\
&\Leftrightarrow c_{xa}^+-\sum_{b=1}^{K}|R^+(x,y,a,b)|=c_{xa}^--\sum_{b=1}^{K}\abs{R^-(x,y,a,b)}\\
&\Leftrightarrow c_{xa}^+-c_{xa}^-=\sum_{b=1}^{K}R^+(x,y,a,b)+\sum_{b=1}^{K}R^-(x,y,a,b)=\sum_{b=1}^{K}R(x,y,a,b)=P(x,a),
\end{align*}which follows from the comments right before Lemma \ref{lemma2}.

For the case $a=K$, we can write
\begin{align*}
&s_{xy}^+(K)=s_{xy}^-(K)\\
&\Leftrightarrow \|R^+\|_\text{NS}-\sum_{a=1}^{K-1}c_{xa}^+-\sum_{b=1}^{K}|R^+(x,y,K,b)|=\|R^-\|_\text{NS}-\sum_{a=1}^{K-1}c_{xa}^--\sum_{b=1}^{K}\abs{R^-(x,y,K,b)}\\
&\Leftrightarrow\|R^+\|_\text{NS}-\|R^-\|_\text{NS}=\sum_{a=1}^{K-1}\big(c_{xa}^+-c_{xa}^-\big)+\sum_{b=1}^{K}R(x,y,K,b)=\sum_{a,b=1}^KR(x,y,a,b)=z.
\end{align*}

Finally, using the same arguments, replacing $a$ with $b$, $x$ with $y$ and $c_{xa}^{\pm}$ with $d_{yb}^{\pm}$, one can show that $t_{xy}^+(b)=t_{xy}^-(b)$ for all $b=1,\cdots, K$.
\end{proof}

\begin{theorem}\label{banachmazurdistance}
The Banach-Mazur distance between $\mathcal{ANS}$ and the space $$BNS_{N,K-1}\oplus_\infty \ell_\infty^N(\ell_1^{K-1}) \oplus_\infty \ell_\infty^N(\ell_1^{K-1})\oplus_\infty \mathbb{R}$$ is upper bounded by 9.
\end{theorem}
\begin{proof}
Define the map T as: 

\begin{align*}
T:\mathcal{ANS}&\rightarrow BNS_{N,K-1}\oplus_\infty \ell_\infty^N(\ell_1^{K-1}) \oplus_\infty \ell_\infty^N(\ell_1^{K-1})\oplus_\infty \mathbb{R}\\
R=&\left\{R(x,y,a,b)\right\}_{a,b=1}^K{}_{x,y=1}^N\\
&\rightarrow \Bigg(\left\{R(x,y,a,b)\right\}_{x,y=1,a,b=1}^{N,K-1}, \left\{\sum_{b=1}^KR(x,y,a,b)\right\}_{x=1,a=1}^{N,K-1},\\
&\ \ \ \ \left\{\sum_{a=1}^KR(x,y,a,b)\right\}_{y=1,b=1}^{N,K-1},\sum_{a,b=1}^KR(x,y,a,b) \Bigg)
\end{align*}

Recall that $\sum_a R(x,y,a,b)$ and also $\sum_b R(x,y,a,b)$ are well defined because $R\in\mathcal{ANS}$ and they do not depend on $x$ or $y$, respectively. Moreover, $\sum_{a,b}R(x,y,a,b)$ is constant for all $x$, $y$. Using these observations, one can easily check that the map $T$ is well defined and it is a linear map. In addition, it is easy to verify that $\|T\|\leq1$. Indeed, this can be seen by noting that the map $T$ can be written as $T=T_1+T_2+T_3+T_4$, where $T_i$ is a linear map and $\|T_i\|\leq 1$ for every $i=1,\cdots, 4$.

The inverse $T^{-1}:BNS_{N,K-1}\oplus_\infty \ell_\infty^N(\ell_1^{K-1}) \oplus_\infty \ell_\infty^N(\ell_1^{K-1}) \oplus_\infty \mathbb{R}\rightarrow \mathcal{ANS}$ of the map $T$ is defined as
\begin{align*}
&T^{-1}\Big(\left\{R(x,y,a,b)\right\}_{x,y=1,a,b=1}^{N,K-1}, \left\{P(x,a)\right\}_{x=1,a=1}^{N,K-1},\left\{Q(y,b)\right\}_{y=1,b=1}^{N,K-1},S \Big)\\&=\tiny\begin{cases} 
      R(x,y,a,b) & \text{if }1\leq a,b\leq K-1 \\
      P(x,a)-\sum_{b'=1}^{K-1}R(x,y,a,b') & \text{if } 1\leq a\leq K-1, b=K\\
      Q(y,b)-\sum_{a'=1}^{K-1}R(x,y,a',b) & \text{if } 1\leq b\leq K-1, a=K\\
      S+\sum_{a',b'=1}^{K-1}R(x,y,a',b')-\sum_{b'=1}^{K-1}Q(y,b')-\sum_{a'=1}^{K-1}P(x,a') & \text{if } a=b=K
   \end{cases}
\end{align*}

Basic linear algebra can be used to show that $T^{-1}$ is well defined; that is, $T^{-1}(R,P,Q,S)=\{O(x,y,a,b)\}_{xyab}$ is in $\mathcal{ANS}$, by showing that $$\sum_{a=1}^K O(x,y,a,b)=Q(y,b)\hspace{0.3 cm}\text{   and   }\hspace{0.3 cm}\sum_{a=1}^K O(x,y,a,K)=S-\sum_{b=1}^{K-1}Q(y,b)\hspace{0.3 cm}\text{   for all   }x,y,b,$$ similar equalities for $\sum_{b=1}^K O(x,y,a,b)$ and also that $\sum_{ab}O(x,y,a,b)=S$ for all $x,y$.

The fact that $T$ is linear is obvious. Finally, to see that $T^{-1}\circ T=id$, call $T^{-1}(T(R))=Z$ and write:
$$
\tiny
Z_{xyab}=\begin{cases} 
      R(x,y,a,b) & \text{if }1\leq a,b\leq K-1 \\
      \sum_{b'=1}^KR(x,y,a,b')-\sum_{b'=1}^{K-1}R(x,y,a,b')=R(x,y,a,K) & \text{if } 1\leq a\leq K-1, b=K\\
      \sum_{a'=1}^KR(x,y,a',b)-\sum_{a'=1}^{K-1}R(x,y,a',b)=R(x,y,K,b) & \text{if } 1\leq b\leq K-1, a=K\\
      \sum_{a',b'=1}^KR(x,y,a',b')+\sum_{a',b'=1}^{K-1}R(x,y,a',b')-\sum_{b'=1}^{K-1}\sum_{a'=1}^KR(x,y,a',b')\\
      -\sum_{a'=1}^{K-1}\sum_{b'=1}^KR(x,y,a',b')=R(x,y,K,K) & \text{if } a=b=K
   \end{cases}
$$

In order to calculate the norm of $T^{-1}$, we can consider four different applications:
\begin{enumerate}
\item[ ] $\alpha_1:BNS_{N,K-1}\rightarrow \mathcal{ANS}$
\item[ ] $\alpha_2:\ell_\infty^N(\ell_1^{K-1}) \rightarrow \mathcal{ANS}$
\item[ ] $\alpha_3:\ell_\infty^N(\ell_1^{K-1}) \rightarrow \mathcal{ANS}$
\item[ ] $\alpha_3:\mathbb{R}\rightarrow \mathcal{ANS}$
\end{enumerate}defined, respectively, by 

\begin{align*}
\alpha_1(R)(x,y,a,b)=\begin{cases}
      R(x,y,a,b) & \text{if }1\leq a,b\leq K-1 \\
      -\sum_{b'=1}^{K-1}R(x,y,a,b') & \text{if } 1\leq a\leq K-1, b=K\\
      -\sum_{a'=1}^{K-1}R(x,y,a',b) & \text{if } 1\leq b\leq K-1, a=K\\
      \sum_{a',b'=1}^{K-1}R(x,y,a',b') & \text{if } a=b=K
   \end{cases}
\end{align*}

\begin{align*}
\alpha_2(P)(x,y,a,b)=\begin{cases} 
      0 & \text{if } 1\leq b\leq K-1\\
      P(x,a) & \text{if } 1\leq a\leq K-1, b=K\\
      -\sum_{a'=1}^{K-1}P(x,a') & \text{if } a=K, b=K
   \end{cases}\end{align*}
   
   \begin{align*}
\alpha_3(Q)(x,y,a,b)=\begin{cases} 
      0 & \text{if } 1\leq a \leq K-1\\
      Q(y,b) & \text{if } 1\leq b\leq K-1, a=K\\
      -\sum_{b'=1}^{K-1}Q(y,b') & \text{if } a=K, b=K
   \end{cases}\end{align*}
   
     \begin{align*}
\alpha_4(S)(x,y,a,b)=\begin{cases} 
      S & \text{if } b=K, a=K\\
      0 & \text{otherwise }
   \end{cases}\end{align*}

One can check that these are well defined linear maps. Moreover, one can write:
$$T^{-1}(R,P,Q,S)=\alpha_1(R)+\alpha_2(P)+\alpha_3(Q)+\alpha_4(S).$$ 

Since $\|\alpha_1(R)\|_{NS}=\max\{\|\alpha_1(R)\|_1,\|\alpha_1(R)\|_2\}$, then,
\begin{align*}
\|\alpha_1(R)\|_1=&\max_x\sum_a\max_y\sum_b \abs{\alpha_1(R)(x,y,a,b)}=\sum_{a,b=1}^K\abs{\alpha_1(R)(x_0,y_a,a,b)}\\
&=\sum_{a,b=1}^{K-1}\abs{R(x_0,y_a,a,b)}+\sum_{a=1}^{K-1}\abs{\sum_{b'=1}^{K-1}R(x_0,y_a,a,b')}+\sum_{b=1}^{K-1}\abs{\sum_{a'=1}^{K-1}R(x_0,y_k,a',b)}\\
&+\abs{\sum_{a',b'=1}^{K-1}R(x_0,y_k,a',b')}\leq 4\|R\|_1.
\end{align*}

Similarly, one can show that $\|\alpha_1(R)\|_1\leq 4\|R\|_1$ making $\|\alpha_1\|\leq4$.

Using the same techniques as before the estimates $\|\alpha_2\|\leq2$, $\|\alpha_3\|\leq2$ and $\|\alpha_4\|\leq1$ can be proven, concluding that 
$$\|T^{-1}\|\leq \|\alpha_1\|+\|\alpha_2\|+\|\alpha_3\|+\|\alpha_4\|\leq 9.$$
\end{proof}

The bound given in Proposition \ref{cota4} (in terms of the number of outputs) was only valid for non-negative elements. We prove now a bound for the general case. First, we need a lemma that allows to bound the difference in norm when changing from $\ell_1^N(\ell_\infty^K(\ell_1^N)=\ell_1^N\otimes_\pi(\ell_\infty^K\otimes_\epsilon\ell_1^N)$ to $\ell_1^N(\ell_\infty^K)\otimes_\epsilon\ell_1^N=(\ell_1^N\otimes_\pi\ell_\infty^K)\otimes_\epsilon\ell_1^N$. Since the result could be of independent interest, we state it and prove it in more general context. The proof remains essentially the same. 
\begin{lemma}\label{cota2bis}
Let $X$ be a Banach space and $M\in \ell_1^N\otimes X\otimes \ell_1^L$. Then,
$$\|M\|_{\ell_1^N(X\otimes_\epsilon\ell_1^L)}\leq \sqrt{2N} \|M\|_{\ell_1^N(X)\otimes_\epsilon\ell_1^L}.$$
\end{lemma}

The proof of this lemma is a consequence of Khintchine inequality (\cite[pag. 96]{defantfloret}), applied to $p=1$, for which we know that $a_1=\sqrt{2}$ in the following result.
\begin{theorem}[Khintchine inequality]\label{Khintchine ineq}
For $1\leq p <\infty$ there exist constants $a_p, \text{}b_p\geq 1$ such that
\begin{align*}
a_p^{-1}\left(\sum_{i=1}^N |\alpha_i|^2\right)^\frac{1}{2}\leq \left(\int_0^1\Big|\sum_{i=1}^N r_i(t)\alpha_i\Big|^p\, dt\right)^{\frac{1}{p}}  \leq b_p\left(\sum_{i=1}^N |\alpha_i|^2\right)^\frac{1}{2}
\end{align*}for every $N$ and all $\alpha_1, \cdots, \alpha_N \in \mathbb R$, where here $(r_i)_{i=1}^N$ denote the Rademacher functions.
\end{theorem}

\begin{proof}
Let $M=\sum_{i=1}^N\sum_{j=1}^L e_i\otimes M_{i,j}\otimes e_j$, with $M_{i,j}\in X$. Then, we have
\begin{align*} 
\|M\|_{\ell_1^N(X\otimes_\epsilon\ell_1^L)}&=\sum_{i=1}^N\Big \| \sum_{j=1}^L M_{i,j}\otimes e_j\Big\|_{X\otimes_\epsilon \ell_1^L}\\&\numeq{1}\sum_{i=1}^N \sup_{x_i^*\in B_{X^*}}\sum_{j=1}^L |x_i^*(M_{i,j})|\\&=\sup_{(x_1^*,\cdots ,x_N^*)\in B_{\ell_\infty^N(X^*)}}\sum_{i=1}^N \sum_{j=1}^L |x_i^*(M_{i,j})|\\&\numineq{2} \sqrt{N}\sup_{(x_1^*,\cdots ,x_N^*)\in B_{\ell_\infty^N(X^*)}}\sum_{j=1}^L \Big(\sum_{i=1}^N |x_i^*(M_{i,j})|^2\Big)^{\frac{1}{2}}\\&\numineq{3} \sqrt{2N}\sup_{(x_1^*,\cdots ,x_N^*)\in B_{\ell_\infty^N(X^*)}}\sum_{j=1}^L \int_0^1 \Big| \sum_{i=1}^N  r_i(t) x_i^*(M_{i,j})\Big| dt\\& =\sqrt{2N}\sup_{(x_1^*,\cdots ,x_N^*)\in B_{\ell_\infty^N(X^*)}}\int_0^1 \sum_{j=1}^L \Big| \sum_{i=1}^N  r_i(t) x_i^*(M_{i,j})\Big| dt\\& \leq \sqrt{2N}\sup_{(x_1^*,\cdots ,x_N^*)\in B_{\ell_\infty^N(X^*)}} \sup_{t\in [0,1]} \sum_{j=1}^L\Big| \sum_{i=1}^N  r_i(t) x_i^*(M_{i,j})\Big|\\&\leq \sqrt{2N}\sup_{(x_1^*,\cdots ,x_N^*)\in B_{\ell_\infty^N(X^*)}} \sup_{(t_1,\cdots, t_N)\in \{-1,1\}^N}\sum_{j=1}^L  \Big|\sum_{i=1}^N  t_ix_i^*(M_{i,j})\Big|\\&\numeq{4} \sqrt{2N}\sup_{(x_1^*,\cdots ,x_N^*)\in B_{\ell_\infty^N(X^*)}} \sum_{j=1}^L\Big| \sum_{i=1}^N  x_i^*(M_{i,j})\Big|\\&\numeq{5}\sqrt{2N}\|M\|_{\ell_1^N(X)\otimes_\epsilon\ell_1^L}.
\end{align*}

Here, $\numeq{1}$ follows from the definition of the $\epsilon$ norm, $\numineq{2}$ follows from the fact that $\|id:\ell_2^N\rightarrow \ell_1^N\|\leq \sqrt{N}$, $\numineq{3}$ follows from Khintchine inequality, $\numeq{4}$ is clear since $\|t_ix_i^*\|=\|x_i^*\|$ for $t_i=\pm 1$ and $\numeq{5}$ follows again from the definition of the $\epsilon$ norm.
\end{proof}

Using this result, we can bound the difference between the norms in $\ell_1^N(\ell_\infty^K)\otimes_\epsilon\ell_1^N(\ell_\infty^K)$ and $\ell_1^N(\ell_\infty^K(\ell_1^N(\ell_\infty^K))$.

\begin{proposition}\label{cota3bis}
There exists a universal constant $C$ independent of $N,K$ such that, 
given $M\in \mathbb {R}^{N^2 K^2}$, one has 
$$
\|M\|_{\ell_1^N(\ell_\infty^K(\ell_1^N(\ell_\infty^K)))}\leq C\sqrt{NK} \|M\|_{\ell_1^N(\ell_\infty^K)\otimes_\epsilon\ell_1^N(\ell_\infty^K)}.$$
\end{proposition}
\begin{proof}
The Banach-Mazur distance between $\ell_\infty^K$ and $\ell_1^K$ is $d(\ell_\infty^K,\ell_1^K)\leq C\sqrt{K}$, with $C$ certain constant independent of the dimension \cite[Proposition 37.6]{tomczak-jaegermann}. This means that there exists an isomorphism $T:\ell_\infty^K\rightarrow \ell_1^K$ such that $\|T\|\|T^{-1}\|\leq C\sqrt{K}$. We will use the metric mapping property of the $\pi$ \cite[pag. 27]{defantfloret} and the $\epsilon$ \cite[pag. 46]{defantfloret} norm,  which says that for all linear maps $T:X\rightarrow W$, $S:Y\rightarrow Z$, we have $$\|T\otimes S:X\otimes_\alpha Y\rightarrow W\otimes_\alpha Z\|=\|T\|\|S\|\hspace{0.2 cm}\text{   for }\hspace{0.2 cm}\alpha=\pi, \epsilon.$$ In particular, if we consider a normed space $X$ and the mapping $id\otimes T: X\otimes_\pi \ell_\infty^K \longrightarrow X\otimes_\pi \ell_1^K$,  then, for every $M\in X\otimes \ell_\infty^K$ one has $\|(id\otimes T)(M)\|_{X\otimes_\pi \ell_1^K}\leq \|T\|\|M\|_{X\otimes_\pi \ell_\infty^K}$. Similar statements hold if we replace $T$ by $T^{-1}$. 

Let $M\in\mathbb{R}^{N^2 K^2}$. The reasonings above, together with  Lemma \ref{cota2bis} replacing the space $X$ in the lemma by $\ell_\infty^K$ yield the following: 
\begin{align*}
\|M\|_{\ell_1^N(\ell_\infty^K(\ell_1^N(\ell_\infty^K)))}&\leq \|T\| \|M\|_{\ell_1^N(\ell_\infty^K(\ell_1^{NK})}\leq \sqrt{2N} \|T\| \|M\|_{\ell_1^N(\ell_\infty^K)\otimes_\epsilon \ell_1^{NK}} \\
&\leq \sqrt{2N} \|T\| \|T^{-1}\|\|M\|_{\ell_1^N(\ell_\infty^K)\otimes_\epsilon \ell_1^{N}(\ell_\infty^K)}\\&\leq 
C'\sqrt{NK} \|\|M\|_{\ell_1^N(\ell_\infty^K)\otimes_\epsilon \ell_1^{N}(\ell_\infty^K)}.
\end{align*}
\end{proof}

\begin{remark}\label{upper bound BNS}
A dual statement of Proposition \ref{cota3bis} is that $$\|id:\ell_\infty^N(\ell_1^K(\ell_\infty^N(\ell_1^K)))\rightarrow \ell_\infty^N(\ell_1^K)\otimes_\pi\ell_\infty^N(\ell_1^K)\|\leq C\sqrt{NK}.$$ Moreover, a dual statement of Proposition \ref{cota1} is $$\|id:\ell_\infty^N(\ell_1^K(\ell_\infty^N(\ell_1^K)))\rightarrow \ell_\infty^N(\ell_1^K)\otimes_\pi\ell_\infty^N(\ell_1^K)\|\leq N.$$
In particular, this trivially implies that $$\|id:BNS_{NK}\rightarrow \ell_\infty^N(\ell_1^{K-1})\otimes_\pi\ell_\infty^N(\ell_1^{K-1})\|\leq C\min\{N, \sqrt{NK}\},$$where the space $BNS_{NK}$ was defined right after Definition \ref{nsnorm}.
\end{remark}

Now we have all the tools to prove Theorem \ref{boundlv}.

\begin{proof}[Proof of Theorem \ref{boundlv}]
As we have said in the introduction, $LV_{\mathcal{NS}}$ is the smallest constant such that $$\widetilde{\mathcal NS}\subseteq LV_{\mathcal{NS}}\cdot \widetilde{\mathcal L},$$where $\widetilde{\mathcal A}=co\big(\mathcal A\cup -\mathcal A\big)$. 

Since the equalities $B_{\mathcal{ANS}}=\widetilde{\mathcal NS}$ and  $B_{NSG\otimes_\pi NSG}=\widetilde{\mathcal L}$ are known from Theorem \ref{NS description} and Theorem \ref{Thm JP}, the statement of the theorem is equivalent to prove that
\begin{align*}
\|id:\mathcal{ANS}\rightarrow NSG\otimes_\pi NSG\|\leq C\text{min}\{N, \sqrt{NK}\},
\end{align*}where $C$ is a universal constant.

Let us define the Banach spaces:
\begin{align*}
&X=BNS_{NK}\oplus_\infty \ell_\infty^N(\ell_1^{K-1})\oplus_\infty \ell_\infty^N(\ell_1^{K-1})\oplus_\infty \R,\\&
Y=\big( \ell_\infty^N(\ell_1^{K-1})\otimes_\pi  \ell_\infty^N(\ell_1^{K-1})\big)\oplus_\infty \ell_\infty^N(\ell_1^{K-1})\oplus_\infty \ell_\infty^N(\ell_1^{K-1})\oplus_\infty \R.
\end{align*}

We will decompose the identity map between $\mathcal{ANS}$ and $NSG\otimes_\pi NSG$ as $$T^{-1}\circ id\circ T:\mathcal{ANS}\rightarrow X\rightarrow Y\rightarrow NSG\otimes_\pi NSG,$$where $T$ is the map used in the proof of Theorem \ref{banachmazurdistance}. Now, in that theorem we showed that $$\|T:\mathcal{ANS}\rightarrow X\|\leq 1.$$Moreover, a direct consequence of Remark \ref{upper bound BNS} is that $$\|id:X\rightarrow Y\|\leq C\min\{N, \sqrt{NK}\}.$$

Hence, we have that $$\|id:\mathcal{ANS}\rightarrow NSG\otimes_\pi NSG\|\leq C\text{min}\{N, \sqrt{NK}\}\|T^{-1}:Y\rightarrow NSG\otimes_\pi NSG\|$$and the theorem will follow from the estimate $$\|T^{-1}:Y\rightarrow NSG\otimes_\pi NSG\|\leq 9.$$

To see this last bound, we proceed as in the proof of Theorem \ref{banachmazurdistance} by decomposing the map $T^{-1}=\alpha_1+\alpha_2+\alpha_3+\alpha_4$ and upper bounding each of the norms. Let us first consider $$\alpha_1:\ell_\infty^N(\ell_1^{K-1})\otimes_\pi  \ell_\infty^N(\ell_1^{K-1})\rightarrow NSG\otimes_\pi NSG.$$Now, in order to upper bound the norm of this map, it suffices to consider elements of the form $R=(P_1(a,x)P_2(b,y))_{x,y;a,b=1}^{N,K-1}$ such that $\|P_1\|_{\ell_\infty^N(\ell_1^{K-1})}\leq 1$ and $\|P_2\|_{\ell_\infty^N(\ell_1^{K-1})}\leq 1$. It is easy to see that $$\alpha_1(R)=\big(Q_1(a,x)Q_2(b,y)\big)_{x,y,a,b=1}^{N,K},$$where for every $x,y$,
$$
Q_1(a|x)=\begin{cases} 
      P_1(a,x) & \text{if } 1\leq a\leq K-1\\
      -\sum_{a'=1}^{K-1}P_1(a,x) & \text{if } a=K.
   \end{cases},$$$$
Q_2(b|y)=\begin{cases} 
      P_2(b,y) & \text{if } 1\leq b\leq K-1\\
      -\sum_{b'=1}^{K-1}P_2(b,y) & \text{if } b=K.
   \end{cases}.$$
   
 Hence, for these particular elements, it is clear that $$\|\alpha_1(R)\|_{NSG\otimes_\pi NSG}=\|Q_1\|_{NSG}\|Q_2\|_{NSG}=\|Q_1\|_{\ell_\infty^N(\ell_1^K)}\|Q_2\|_{\ell_\infty^N(\ell_1^K)},$$where in the last equality we have used Lemma \ref{NSG-L1} \footnote{Note that, as we mentioned in Remark \ref{Norm-complementation}, in general we cannot replace $\|\alpha_1(R)\|_{NSG\otimes_\pi NSG}$ by $\|\alpha_1(R)\|_{\ell_\infty^N(\ell_1^K)\otimes_\pi \ell_\infty^N(\ell_1^K)}$. However, for the particular elements of the form $Q_1\otimes Q_2$, both norms coincide by Lemma \ref{NSG-L1}.}. Now, it is very easy to check that  $\|Q_1\|_{\ell_\infty^N(\ell_1^K)}\leq 2$ and $\|Q_2\|_{\ell_\infty^N(\ell_1^K)}\leq 2$, from where we conclude that $\|\alpha_1\|\leq 4$.

Let us consider now $$\alpha_2:\ell_\infty^N(\ell_1^{K-1}) \rightarrow NSG\otimes_\pi NSG.$$Given $P\in \ell_\infty^N(\ell_1^{K-1})$ with $\|P\|_{\ell_\infty^N(\ell_1^{K-1})}\leq 1$, it can be easily checked that 
$$\alpha_2(P)=\big(Q_1(a,x)Q_2(b,y)\big)_{x,y,a,b=1}^{N,K},$$where for every $x,y$,
\begin{align*}
&Q_1(a|x)=\begin{cases} 
      P(a,x) & \text{if } 1\leq a\leq K-1\\
      -\sum_{a'=1}^{K-1}P_1(a',x) & \text{if } a=K
   \end{cases},\\
&Q_2(b|y)=\begin{cases} 
      0 & \text{if } 1\leq b\leq K-1\\
      1 & \text{if } b=K
   \end{cases}.\end{align*}
   
 As in the case if $\alpha_1$, we can deduce that $\|\alpha_2(P)\|_{NSG\otimes_\pi NSG}=\|Q_1\|_{\ell_\infty^N(\ell_1^K)}\|Q_2\|_{\ell_\infty^N(\ell_1^K)}\leq 2$, so that $\|\alpha_2\|\leq 2$. Moreover, the case of $\alpha_3$ can be analyzed exactly in the same way to deduce $\|\alpha_3\|\leq 2$. 
 
 Finally,  for the case of $\alpha_4:\mathbb R \rightarrow NSG\otimes_\pi NSG$, one can check that for a given $|s|\leq 1$, we have $\alpha_4(s)=\big(Q_1(a,x)Q_2(b,y)\big)_{x,y,a,b=1}^{N,K}$, where ,for every $x,y$,
$$
Q_1(a|x)=\begin{cases} 
      0 & \text{if } 1\leq a\leq K-1\\
      s & \text{if } a=K.
   \end{cases},$$$$
Q_2(b|y)=\begin{cases} 
      0 & \text{if } 1\leq b\leq K-1\\
      1 & \text{if } b=K.
   \end{cases},$$and on trivially deduces that $\|\alpha_4\|\leq 1$.
   
  Since $\|T^{-1}\|\leq \|\alpha_1\|+\|\alpha_2\|+\|\alpha_3\|+\|\alpha_4\|\leq 9$, we conlude the proof.
\end{proof}

\begin{remark}
For this case it can also be seen that using the same techniques, when we distinguish the inputs and outputs for Alice and Bob as $N_1$, $N_2$, $K_1$ and $K_2$, the following bounds can be obtained: 
$$LV_{\mathcal{NS}}(M)\leq\mathcal{O}(\min\{N_1,N_2,\sqrt{CN_1K_2},\sqrt{CN_2K_1}\})$$
\end{remark}

\section*{acknowledgment}
This research was funded by the Spanish MINECO through Grant No. MTM2017-88385-P and by the Comunidad de Madrid through grant QUITEMAD-CM P2018/TCS4342. We also acknowledge funding from SEV-2015-0554-16-3 and ``Ram\'on y Cajal program'' RYC-2012-10449 (C. P.).


\begin{thebibliography}{20}

\bibitem{ABGMPS} A. Ac\'in, N. Brunner, N. Gisin, S. Massar, S. Pironio, V. Scarani, \emph{Device-independent security of
quantum cryptography against collective attacks}, Phys. Rev. Lett. 98, 230501 (2007).


\bibitem{BS15} M. Bavarian, P. W. Shor, \emph{Information causality, Szemer\'{e}di-Trotter and algebraic variants of CHSH}, Proceedings of the 2015 Conference on Innovations in Theoretical Computer Science, 123-132 (2015).

\bibitem{Bell} J.S. Bell, \emph{On the Einstein-Poldolsky-Rosen paradox}, Physics, \textbf{1}, 195 (1964).

\bibitem{interpolation} J. Bergh, J. L{\"o}fstr{\"o}m, \emph{Interpolation Spaces, An Introduction}, Grundlehren der mathematischen Wissenschaften, 223. Berlin-Heildelberg-New York, Springer-Verlag (1976).

\bibitem{BCPSW} N. Brunner, D. Cavalcanti, S. Pironio, V. Scarani, S. Wehner, \emph{Bell nonlocality}, Rev. Mod. Phys. 86, 419 (2014).

\bibitem{BuhrmanRSW12} H. Buhrman, O. Regev, G. Scarpa, and R. de~Wolf, Near-optimal and explicit {B}ell inequality violations  {\em Theory Comput.}, 8:623-645(2012).

\bibitem{defantfloret} A. Defant and K. Floret, \emph{Tensor Norms and Operator Ideals}, North-Holland, Amsterdam (1993)

\bibitem{communicationcomplexityns} J. Degorre, M. Kaplan, S. Laplante and J. Roland, \emph{The Communication Complexity of non-signalling Distributions}, Mathematical Foundations of Computer Science 270-281(2009).

\bibitem{FL92} U. Feige and L. Lovasz, \emph{Two-prover one-round proof systems: Their power and their problem}. In Proceedings of the 24th ACMSymposium on Theory of Computing, 733-741 (1992).

\bibitem{chernoff} M. Goemans, \emph{Chernoff bounds, and some applications} 18.310 lecture notes, \url{http://math.mit.edu/~goemans/18310S15/chernoff-notes.pdf} (2015).

\bibitem{khintchine} U. Haagerup, \emph{The best constants in the Khintchine inequality}, Studia Math. 70 (1981), no. 3, 231-283 (1982)

\bibitem{ito}T. Ito, \emph{Polynomial-Space Approximation of No-Signaling Provers} "Automata, Languages and Programming" Volume 6198 of Lecture Notes in Computer Science 140-151 (2010).

\bibitem{JungeP11low} M. Junge, C. Palazuelos, Large violation of Bell inequalities with low entanglement, \emph{Comm. Math. Phys.}, 1-52 (2011).

\bibitem{JPPVW}M. Junge, C. Palazuelos, D. P\'erez-Garc\'ia, I. Villanueva and M.M. Wolf, \emph{Unbounded violations of bipartite Bell Inequalities via Operator Space theory}, Comm. Math. Phys., 300, 715-739 (2010).

\bibitem{JPPVW2}M. Junge, C. Palazuelos, D. P\'erez-Garc\'ia, I. Villanueva and M.M. Wolf, \emph{Operator Space theory: a natural framework for Bell inequalities}, Phys. Rev. Lett. 104, 170405 (2010).


\bibitem{LPW18} L. Lami, C. Palazuelos and A. Winter, \emph{Ultimate Data Hiding in Quantum Mechanics and Beyond}, Commun. Math. Phys. 361(2):661-708 (2018)

\bibitem{LW} C. Lancien and A. Winter, \emph{Parallel repetition and concentration for (sub-)no-signalling games via a flexible constrained de Finetti reduction}, Chicago J. Theor. Comput. Sci. (2016)

\bibitem{PV-Survey} C. Palazuelos, T. Vidick, \emph{Survey on Nonlocal Games and Operator Space Theory}. J. Math. Phys. 57, 015220 (2016).


\bibitem{tomczak-jaegermann}N. Tomczak-Jaegermann, \emph{Banach-Mazur Distances and Finite Dimensional Operator Ideals}, Longman Scientific \& Technical, 1989

\bibitem{TsirelsonII} B. S. Tsirelson, \emph{Quantum analogues of the Bell inequalities. The case of two spatially separated domains}, Journal of Soviet Mathematics, 36(4): 557-570 (1987).

\bibitem{Tsirelson} B. S. Tsirelson, \emph{Some results and problems on quantum Bell-type inequalities}, Hadronic J. Supp. 8(4), 329-345 (1993).

\end{thebibliography}
\end{document}